\documentclass{IEEEtran}
\usepackage{amsmath,epsfig,cite}
\usepackage{amssymb, amsthm}
\usepackage{mathrsfs}
\usepackage{algorithm,algorithmic}
\usepackage{tabularx}
\usepackage{threeparttable}

\newtheorem{myth}{\bf Theorem}

\newtheorem{myle}[myth]{\bf Lemma}

\theoremstyle{definition}
\newtheorem{mydef}{\bf Definition}
\newtheorem{myexam}{\bf Example}

\newcommand{\bx}{\boldsymbol x}

\newcommand{\bzero}{\boldsymbol 0}
\newcommand{\bone}{\boldsymbol 1}

\newcommand{\IEVq}{$\mathsf{IEV}_q$}
\newcommand{\IEV}{$\mathsf{IEV}_2$}

\newcommand{\MAXIEVq}{$\mathsf{Max}$-$\mathsf{IEV}_q$}

\begin{document}

\graphicspath{{figures/}}

\title{Linear Network Code for Erasure Broadcast Channel with Feedback: Complexity and Algorithms}

\author{\large Chi Wan Sung, Linyu Huang, Ho Yuet Kwan, and Kenneth W. Shum}

\date{}

%
\long\def\symbolfootnote[#1]#2{\begingroup
\def\thefootnote{\fnsymbol{footnote}}\footnote[#1]{#2}\endgroup}

\maketitle

\begin{abstract}
 This paper investigates the construction of linear network codes for broadcasting a set of data packets to a number of users. The links from the source to the users are modeled as independent erasure channels. Users are allowed to inform the source node whether a packet is received correctly via feedback channels. In order to minimize the number of packet transmissions until all users have received all packets successfully, it is necessary that a data packet, if successfully received by a user, can increase the dimension of the vector space spanned by the encoding vectors he or she has received by one. Such an encoding vector is called innovative. We prove that innovative linear network code is uniformly optimal in minimizing user  download delay. On the other hand, innovative encoding vectors do not always exist. When the finite field size is strictly smaller than the number of users, the problem of determining the existence of innovative vectors is proven to be NP-complete. When the field size is larger than or equal to the number of users, innovative vectors always exist and random linear network code (RLNC) is able to find an innovative vector with high probability. While RLNC is optimal in terms of completion time, it has high decoding complexity due to the need of solving a system of linear equations. To reduce decoding time, we propose the use of sparse linear network code, since the sparsity property of encoding vectors can be exploited when solving systems of linear equations. Generating a sparsest encoding vector with large finite field size, however, is shown to be NP-hard. An approximation algorithm that guarantee the Hamming weight of a generated encoding vector to be smaller than a certain factor of the optimal value is constructed. For the binary field, a heuristic algorithm is also proposed. Our simulation results show that our proposed methods have excellent performance in completion time and outperforms RLNC in terms of decoding time. This improvement is obtained at the expense of higher encoding time and the requirement of user feedback. Comparisons with other broadcast codes have been made and numerical results show that different tradeoff can be obtained by these schemes.
\end{abstract}


\symbolfootnote[0]{This work was supported in part by a grant from the University Grants Committee of the Hong Kong Special Administrative Region, China under Project AoE/E-02/08, and in part by a grant from the Research Grants Council of the Hong Kong Special Administrative Region, China under Project CityU 121713.}
\symbolfootnote[0]{Chi Wan Sung and Linyu Huang are with Department of Electronic Engineering, City University of Hong Kong, Tat Chee Ave, Kowloon Tong, Hong Kong. Email: albert.sung@cityu.edu.hk, l.huang@my.cityu.edu.hk}
\symbolfootnote[0]{Ho Yuet Kwan is with Division of Applied Science and Technology, Community College of City University, Tat Chee Ave, Kowloon Tong, Hong Kong. Email: hykwan@cityu.edu.hk.}
\symbolfootnote[0]{Kenneth Shum is with Institute of Network Coding, The Chinese University of Hong Kong, Shatin, Hong Kong. Email: wkshum@inc.cuhk.edu.hk.}
\symbolfootnote[0]{This paper was  presented in part in Int. Symp. on Inform. Theory, St. Petersburg, 2011,  and in Int. Symp. in Network Coding, Beijing,  2011.}

\section{Introduction}

Broadcasting has been a challenging issue in telecommunications. The challenge mainly comes from how a transmitter can disseminate a common information content to all users/receivers reliably and efficiently via a broadcast channel which could be unstable and error-prone. More specifically, one of the ultimate goal of broadcasting is to provide a transmission scheme such that a common information content or a set of packets can be disseminated with minimum number of transmissions for a sender to complete the whole information content dissemination for all users. This measure is commonly called the {\em completion time} of a broadcast system.

Several classical approaches provide heuristic solutions to the above issue. With user feedback, automatic repeat request (ARQ) offers reliable retransmissions for the erased packets due to channel impairments. However, such an approach becomes inefficient when the number of users increases, as the users may have entirely distinct needs for the erased packets. Reliable broadcasting can also be achieved without user feedback by forward error correction. With the use of erasure codes, a user can reconstruct the entire set of original packets, provided that the number of erased packets is smaller than a certain threshold. However, the amount of packets that will be erased depends on the channel erasure probability, which is time-varying and hard to predict. That limits the use of erasure codes in broadcasting. To improve upon classical approaches, the approach of linear network coding \cite{LNC1,LNC2} has been shown to be a promising solution \cite{EOM06,GTK08,HPFL09,DFK09}.


The idea of linear network coding for broadcasting is that a transmitter broadcasts to $K$ users encoded packets that are obtained by linearly combining the $N$ original packets over the finite field $GF(q)$. An encoding vector specifies the coefficients for the linear combination. An encoded packet together with a header which contains the corresponding encoding vector is broadcasted to all users. It is said to be {\em innovative to a user} if the corresponding encoding vector is not in the subspace spanned by the encoding vectors already received by that user. It is called {\em innovative} if it is innovative to all users who have not yet received enough packets for decoding. It is shown in \cite{DFT07} that an innovative packet can always be found if $q \geq K$. Once a user receives any $N$ innovative packets, he or she can decode the $N$ original packets by Gauss-Jordan elimination. Therefore, the generation of innovative packets is vital. Clearly, if all the encoded packets are innovative, the completion time can then be minimized.

Linear network codes for broadcasting can be generated with or without feedback. LT code \cite{LT}, Raptor code \cite{Raptor} and random linear network code (RLNC) \cite{Ho} can be used without feedback. By suitably choosing design parameters, innovative packets can be generated by those coding schemes with high probability. LT code and Raptor code are generated by an optimized degree distribution. However, they are mainly designed for broadcasting a huge number of packets, and may not be good choices when the number of packets is only moderately large. With feedback, it is suggested in \cite{DFT07} the use of Jaggi-Sanders algorithm \cite{JSCEEJT}, which is a general network code generation method and is able to find innovative encoding vectors for $q\geq K$. However, its encoding and decoding complexities are relatively high, as it is not specifically designed for the broadcast application. Therefore, some heuristics have been proposed \cite{SSM08,SSM09, NNB09,Xiao2010}. It is suggested in \cite{ID2} that encoded packets should be {\em instantly decodable}, in the sense that a new packet can be decoded once it is available at a receiver without waiting for the complete reception of the full set of packets. However, as an instantly decodable packet to all users may not exist, the completion time is in general larger than that in a system without this extra requirement. With the idea of instantly decodability, some works focus on minimizing {\em decoding delay}, where a unit of decoding delay is defined as that an encoded packet is successfully received by a user but that packet is not innovative or not instantly decodable to him or her \cite{KDF08,CMWB,ID1,SST10}.

The excellent performance of linear network coding to broadcast encourages researchers to consider its practicality. In fact, the decoding complexity of a linear network code is an important issue in practice. One possible way to reduce decoding complexity is to use sparse encoding vectors. This sparsity property is important, as it can be exploited in the decoding process. For example, a fast algorithm by Wiedemann for solving a system of sparse linear equations can be used~\cite{Wiedemann86}. If the Hamming weight of each encoding vector is at most $w$, the complexity for solving an $N\times N$ linear system can be reduced from $O(N^3)$ using Gaussian elimination to $O(wN^2)$ \cite{sungonthesparsity2011}. The Wiedemann algorithm is useful when $N$ is large. When $N$ is moderate, we can implement some sparse representation of matrices, so that even if the usual Gaussian elimination is used, the number of additions and multiplications required can be reduced. For other fast methods for solving linear equations over finite fields, we refer the readers to~\cite{Kaltofen,Coppersmith}.

Minimizing the completion time and reducing the decoding complexity are equally important in linear network code design for erasure broadcast channel. However, the innovativeness of encoding vectors together with their sparsity has not been thoroughly studied. Given the encoding vectors which have been received by the users in a broadcast system, the generation of encoding vectors which are both sparse and innovative is a challenging problem. In this paper, we address the issue by developing a method called the {\em Optimal Hitting Method} and its approximation version, called the {\em Greedy Hitting Method}. Both of them are able to generate sparse and innovative encoding vectors for $q \geq K$. That results in a significant reduction in decoding complexity when compared with their non-sparse counterparts. Furthermore, based on the greedy hitting method, we develop a suboptimal procedure to improve the completion time performance for $q=2$, where the existence of innovative encoding vectors is not guaranteed. Simulation results show that its performance is nearly optimal.

The rest of this paper is organized as follows. We review the literature on complexity in network coding in Section~\ref{sec:related} and some useful notions in complexity theory in Section~\ref{sec:notions}. In Section~\ref{sec:formulation}, the system model is introduced and the problem is formulated.
In Section~\ref{sec:lnc}, we show that innovative linear network code is uniformly optimal. In Section~\ref{sec:secinno}, we characterize innovative encoding vectors by a linear algebra approach and prove that the determination of the existence of an innovative vector for $q<K$ is NP-complete. In Section~\ref{sec:sparsity}, the sparsity issue is considered. After showing that $K$-sparse innovative vectors always exist if $q \geq K$, we investigate the {\sc Sparsity} problem and prove that it is NP-complete. In Section~\ref{sec:IVGA}, we present a systematic way to solve {\sc Max Sparsity} using binary integer programming. A polynomial-time approximation algorithm is also constructed. In Section~\ref{sec:benchmark}, some benchmark algorithms for wireless broadcast are described. In Section~\ref{sec:PE}, our algorithms are compared with those benchmarks by simulations. Finally, conclusions are drawn in Section~\ref{sec:conclude}.

\section{Literature on Complexity Classes of Network Coding Problems} \label{sec:related}

A considerable amount of research has been done on the complexity issues in conventional coding theory (See the survey in~\cite{Barg} for example). For instance, it is shown in \cite{BMT78} and \cite{Vardy97} that the problems of finding the weight distribution and the minimum distance of linear codes are NP-hard. The complexity issues in network coding are less well understood.

For linear network codes, Lehman and Lehman investigated the complexity of a class of network coding problems in~\cite{Lehman}, and proved that some of the problems are NP-complete. Construction of linear network codes using a technique called matrix completion is considered in~\cite{HKM05}, and the complexity class of the matrix completion problem is studied in~\cite{HKY06}. It is shown in~\cite{LS11} that approximating the capacity of network coding is also a hard problem.

To minimize encoding complexity, Langberg, Sprintson and Bruck divide the nodes in a general network topology into two classes. The nodes in one class forward packets without any coding while the nodes in another class perform network coding. The problem of minimizing the number of encoding nodes is shown to be NP-complete in~\cite{LSJ06, LSJ09}.

El Rouayheb, Chaudhry and Sprintson study the complexity of a related problem called {\em index coding problem} in~\cite{ERCS07}. They consider the noiseless broadcast channel, and show that when the coefficient field is binary, the problem of minimizing the number of packet transmissions is NP-hard. A complementary version of the index coding is studied in~\cite{CASL}. It is shown that the complementary index coding is NP-hard, and even obtaining an approximate solution is NP-hard.

In~\cite{MPRG}, Milosavljevic {\em et al.} studies a related system. The users are interested in a common data file but only have partial knowledge of the file. By interactively sending data to each others through a noiseless broadcast channel, the users want to minimize the total amount of data sent through the channel. It is shown in~\cite{MPRG} that the optimal rate allocations can be found in polynomial time.

In this paper, the problem setting is similar, except that the channel is modeled as an erasure broadcast channel, and we focus on the innovativeness and sparsity aspects of generating encoding vectors.

\section{Useful Notions in Complexity Theory} \label{sec:notions}

Before presenting the broadcast problem, we first define some useful notions in complexity theory, which will be used in this paper. The following definitions are taken from~\cite{goldreich}:

\begin{mydef}
Let $\{0,1\}^*$ denote the set of all binary strings, and
$S$ be a subset of $\{0, 1\}^*$. A function $f : \{0, 1\}^* \rightarrow \{0,1\}$ is said to solve the {\em decision problem} of $S$ if for every binary string $x$ it holds that $f(x) = 1$ if and only if $x \in S$.
\end{mydef}

\begin{mydef}
For a given $R \subseteq\{0,1\}^* \times \{0, 1\}^*$, let $R(x) \triangleq \{y : (x,y) \in R\}$ denote the set of solutions for the binary string $x$. A function $f : \{0,1\}^* \rightarrow \{0,1\}^* \cup \{\perp\}$ is said to solve the {\em search problem} of $R$ if for every $x$ the following holds:
\[
 f(x) \begin{cases}
 \in R(x) & \text{if } R(x) \neq \emptyset, \\
  =\ \perp & \text{otherwise}.
 \end{cases}
\]
\end{mydef}

Note that a minimization problem can be regarded as a search problem. By definition, a minimization problem is associated with a value function $V: \{0,1\}^* \times \{0,1\}^* \rightarrow \mathbb{R}$. Given $x$, the task is to find $y$ such that $(x,y) \in R$ and $V(x,y)$ is the minimum value of $V(x,y')$ for all $y' \in R(x)$.

The following two definitions concerning reductions between two problems:

\begin{mydef}
For $S$, $S'\subseteq \{0,1\}^*$,
A polynomial-time computable function $f:S\rightarrow S'$ is called a {\em Karp-reduction} of $S$ to $S'$ if, for every binary string $x$, it holds that $x \in S$ if and only if $f(x) \in S'$.
\end{mydef}

\begin{mydef} \label{def:levin}
For $R$, $R' \subseteq \{0,1\}^*\times \{0,1\}^*$,
a pair of polynomial-time computable functions,
\begin{align*}
f & : \{0,1\}^* \rightarrow \{0,1\}^*, \\
g & : \{0,1\}^* \times \{0,1\}^* \rightarrow \{0,1\}^*,
\end{align*}
is called a {\em Levin-reduction} of $R$ to $R'$ if the function $f$ is a Karp-reduction of $S_R \triangleq \{x : \exists y \text{ s.t. } (x,y) \in R\}$ to $S_{R'} \triangleq \{x' : \exists y' \text{ s.t. } (x', y') \in R'\}$, and for every $x\in S_R$ and $ (f(x),y') \in R'$ it holds that $(x,g(x,y')) \in R$.
\end{mydef}

Detailed explanations of the above concepts can be found in~\cite{goldreich}.

\section{System Model and Problem Formulation} \label{sec:formulation}

Consider a single-hop wireless broadcast system, in which there are one source and $K$ users. The source wants to send $N$ data packets to all the $K$ users. We view each packet as a symbol from an alphabet set $\mathcal{X}$ of size $q$. In other words, the source wants to broadcast $N$ symbols, $P_1, P_2, \ldots, P_N \in \mathcal{X}$. We assume that they are independent random variables, each of which is drawn uniformly at random from $\mathcal{X}$.

We model the transmission as a time-slotted broadcast erasure channel. In time slot~$t$, a symbol $X_t \in \mathcal{X}$ is transmitted by the source. The channel output observed by user~$k$, denoted by $Y_{k,t}$, is either the same as $X_t$ or equal to a special erasure symbol $e$.
A time slot is called a non-erasure slot of user $k$ at time $t$ if $X_t = Y_{k,t}$, and is called an erasure slot of user $k$ otherwise.
The channel dynamics is modeled by a stochastic sequence,
\[\Psi \triangleq \big( (S_{1,t}, S_{2,t}, \ldots, S_{K,t}) \big)_{t=1,2,3,\ldots},\]
where $S_{k,t}$ equals one if $Y_{k,t} = X_t$ or zero if $Y_{k,t} = e$. Assume that $S_{k,t}$'s are all independent of the source symbols $P_1, P_2, \ldots, P_N$. For $k = 1, 2, \ldots, K$, we let $N_k(t,\Psi)$ be the number of non-erasure slots of user~$k$ in the first $t$ time slots. We let $\Psi_T$ be the truncated sequence obtained from $\Psi$ by preserving the $K$ random variables in the first $T$ time slots. After every slot~$t$, user~$k$ broadcasts $S_{k,t}$ via a control channel without delay and error. We assume that after time $\tau$, the source and all users have the knowledge of $\Psi_\tau$.

Define $\mathcal{Y} \triangleq \mathcal{X} \cup \{e\}$.
An $(N,K,q)$ broadcast code is defined by encoding functions
\begin{equation}
f_t: \mathcal{X}^N \times \{0,1\}^{K(t-1)} \rightarrow \mathcal{X},
\end{equation}
and decoding functions
\begin{equation} \label{decode_def}
g_{k,t}: \mathcal{Y}^t \times \{0,1\}^{Kt} \rightarrow \mathcal{X}^N,
\end{equation}
where $k = 1, 2, \ldots, K$ and $t = 1, 2, \ldots$.

Given a broadcast code and a realization of the channel dynamics $\Psi$, user~$k$ is said to have {\em download delay}~$T_k(\Psi)$ if it is the smallest value of $t$ such that decoding is successful, that is,
\begin{equation} \label{decoding}
g_{k,t}(Y_{k,1}, Y_{k,2}, \ldots, Y_{k,t}, \Psi_t) = (P_1, P_2, \ldots, P_N).
\end{equation}
If decoding is never successful, then we let the download delay be infinity.

The following result gives a lower bound of the download delay of each user:

\begin{myth}
Given any $(N,K,q)$ broadcast code and any channel realization $\Psi$, we have
$T_k(\Psi) > \tau$ for all $\tau$ such that $N_k(\tau,\Psi)<N$, for all $k = 1, 2, \ldots, K$.
\end{myth}

\begin{proof}
Consider a time index $\tau$, where $N_k(\tau,\Psi) < N$. Let $a_1, a_2, \ldots, a_{N_k(\tau)} \leq \tau$ be the indices of time slots at which user~$k$ experiences no erasure, and let $Y_k \triangleq (Y_{k,a_1}, Y_{k,a_2}, \ldots, Y_{k,a_{N_k(\tau)}})$. Note that
\begin{align*}
& H(P_1, P_2, \ldots, P_N | Y_{k,1}, Y_{k,2}, \ldots, Y_{k,\tau}, \Psi_\tau) \\
&= H(P_1, P_2, \ldots, P_N | Y_k) \\
&= H(P_1, P_2, \ldots, P_N) - \big[ H(Y_k) - H(Y_k | P_1, P_2, \ldots, P_N) \big] \\
&\geq H(P_1, P_2, \ldots, P_N) - H(Y_k) \\
&= N \log_2 |\mathcal{X}| - H(Y_k) \\
&\geq N \log_2 |\mathcal{X}| - N_k(\tau,\Psi) \log_2 |\mathcal{X}| \\
&= (N- N_k(\tau,\Psi)) \log_2 q \\
&> 0
\end{align*}
Therefore, the probability that the decoding condition in~\eqref{decoding} holds must be strictly less than one. In other words, the download delay of user~$k$, $T_k(\Psi)$, must be strictly greater than~$\tau$, for all $k$'s.
\end{proof}

\begin{mydef}
An $(N,K,q)$ broadcast code is said to be uniformly optimal if for any channel realization $\Psi$ and $k = 1, 2, \ldots, K$,
\begin{equation} \label{uoptimal}
T_k(\Psi) = \min \{ \tau : N_k(\tau, \Psi) = N \}.
\end{equation}
If the minimum does not exist, we define it as infinity.
\end{mydef}

The existence of uniformly optimal broadcast code will be investigated in the next two sections.

\section{Linear Network Code} \label{sec:lnc}

In this paper, we focus on the use of linear network code. The alphabet set $\mathcal{X}$ is identified with the finite field $GF(q)$ of size $q$, for some prime power~$q$. We define linear network code formally below:

\begin{mydef}
An $(N,K,q)$ broadcast code is said to be a linear network code if its encoding functions can be expressed as a linear function of the source packets:
\begin{equation} \label{LNC}
f_t(P_1, P_2, \ldots, P_N, \Psi_{t-1}) = x_1 P_1 + x_2 P_2 + \cdots + x_N P_N,
\end{equation}
where $x_1, x_2, \ldots, x_N \in \mathcal{X}$ are determined by $\Psi_{t-1}$, and the addition and multiplication operations are defined over $GF(q)$.
\end{mydef}

The vector $\mathbf{x} \triangleq (x_1, x_2, \ldots, x_N) \in GF(q)^N$, as expressed in~\eqref{LNC}, is called the {\em encoding vector} of the packet transmitted in slot~$t$. Throughout this paper, all vectors are assumed to be column vectors, and we use parenthesis and commas when its components are listed horizontally.

For practical applications, the transmitter can put the encoding vector in the header of the encoded packet. While that incurs some transmission overhead, it can relax the requirement specified in the previous section that every user can listen to the feedback information from all other users. In other words, the decoding function of user $k$ in~\eqref{decode_def} can be changed to
\begin{equation}
g_{k,t}: \mathcal{Y}^t \times GF(q)^{Nt} \rightarrow \mathcal{X}^N,
\end{equation}
assuming that the decoder knows the encoding vectors of its received packets.

The {\em support} of the vector ${\bf x}$, denoted by $supp({\bf x})$, is the set of indices of the non-zero components in ${\bf x}$, i.e.,
\[
 supp({\bf x}) \triangleq \{i :\, x_i \neq 0 \}.
\]
The {\em Hamming weight} of ${\bf x}$ is defined as the cardinality of  $supp({\bf x})$. An encoding vector that has Hamming weight less than or equal to $w$ is said to be $w$-sparse.

Note that a transmitted packet brings new information to a user if and only if its entropy conditioned on the perviously received packets by that user is greater than zero, or equivalently, the new packet is not a function of the previously received packet. In linear algebraic terms, the condition is that the encoding vector of the new packet does not lie within the span of all previously received encoding vectors of that user. We say that such an encoding vector is {\em innovative to that user}. An encoding vector that is innovative to all users is simply said to be {\em innovative}.

Suppose that user~$k$, for $k=1,2,\ldots,K$, has already received $r_k$ packets whose encoding vectors are linearly independent. Let $\mathbf{C}_k$ be the $r_k \times N$ encoding matrix of user~$k$, whose rows are the transposes of the $r_k$ encoding vectors. Without loss of generality, we assume that $r_k < N$, for otherwise user~$k$ can decode the file successfully and can be omitted from our consideration. A vector $\mathbf{x}$ is innovative if it does not belong to the row space of $\mathbf{C}_k$ for any $k$. Given $K$ encoding matrices $\mathbf{C}_1, \mathbf{C}_2, \ldots, \mathbf{C}_K$, the set of all innovative encoding vectors, ${\cal I}$, is given by
\begin{equation}
\mathcal{I} \triangleq GF(q)^N \setminus \bigcup_{k=1}^K \text{rowspace}(\mathbf{C}_k). \label{inno}
\end{equation}


\begin{mydef}
A linear network code is said to be innovative if for any channel realization $\Psi$, its encoded packet at time~$t$ is innovative to all users who have not successfully decoded the source packets yet, that is, those users with indices in $\{ k : T_k(\Psi) \geq t \}$.
\end{mydef}

\begin{myth}
Innovative linear network codes are uniformly optimal.
\end{myth}

\begin{proof}
With an innovative linear network code, by definition, the packets received by a user who has not successfully decoded all the source packets must all be linearly independent. Therefore, she is able to decode the source packets once she has experienced
$N$ non-erasure slots. In other words,~\eqref{uoptimal} holds for all users. Hence the code is uniformly optimal.
\end{proof}

In the next section, we will show that innovative linear network codes exist when $q \geq K$.

\section{The Innovative Encoding Vector Problem} \label{sec:secinno}

The existence of innovative linear network code is equivalent to the non-emptiness of the set of encoding vectors $\mathcal{I}$ as defined in~\eqref{inno}. It was shown in~\cite{DFT07} that ${\cal I}$ is non-empty if the finite field size, $q$, is larger than or equal to the number of users, $K$. We present a proof below for the sake of completeness. We begin with a simple lemma, which will be used again in a later section.

\begin{myle}
Let $\mathcal{A}_1, \mathcal{A}_2, \ldots, \mathcal{A}_K$ be finite subsets of a universal set $\mathcal{U}$.  If $K\geq 2$ and $\mathcal{A}_1, \mathcal{A}_2, \ldots, \mathcal{A}_K$ contain a common element, then
\[
|\mathcal{A}_1 \cup \mathcal{A}_2 \cup \cdots \cup \mathcal{A}_K |  < |\mathcal{A}_1| + |\mathcal{A}_2| + \cdots +|\mathcal{A}_K|.
\]
\label{special_case_of_Schwartz}
\end{myle}

\begin{proof}
Suppose $x \in \mathcal{A}_i$ for all $i$. Let $\mathcal{A}^*_i$ be the set $\mathcal{A}_i \setminus \{x\}$  for $i=1,2,\ldots,K$. By applying the union bound,
\[
|\mathcal{A}^*_1 \cup \mathcal{A}^*_2 \cup \cdots \cup \mathcal{A}^*_K |  \leq |\mathcal{A}^*_1| + |\mathcal{A}^*_2| + \cdots +|\mathcal{A}^*_K|.
\]
This implies that
\[
|\mathcal{A}_1 \cup \mathcal{A}_2 \cup \cdots \cup \mathcal{A}_K | -1  \leq |\mathcal{A}_1| -1 + |\mathcal{A}_2| -1 + \cdots +|\mathcal{A}_K| -1.
\]
As $K\geq 2$ by hypothesis, we obtain the inequality in the lemma.
\end{proof}

\begin{myth}[\cite{DFT07}] \label{nonemptyI}
If $q\geq K$ and the rank of $\mathbf{C}_k$ is strictly less than $N$ for all $k$'s, then $\mathcal{I}$ is non-empty.
\end{myth}

\begin{proof}
For $k=1,2,\ldots, K$, let $V_k$ be the row space of $\mathbf{C}_k$. The subspace $V_k$ consists of the $q^{r_k}$ encoding vectors which are {\em not} innovative to user $k$. Obviously the zero vector is a common vector of these $K$ subspaces. By Lemma~\ref{special_case_of_Schwartz}, the union of these $K$ subspaces contains strictly less than $\sum_{k=1}^K q^{r_k}$ vectors. Since $K\leq q$, we have $\sum_{k=1}^K q^{r_k} \leq K q^{N-1} \leq  q^N$. Therefore there exists at least one encoding vector which is innovative to all users.
\end{proof}

The condition $q \geq K$ in Theorem~\ref{nonemptyI} cannot be improved in general. In Appendix~\ref{app:A} we construct examples with no innovative encoding vector for $K = q+1$.

\smallskip

The set of innovative encoding vectors, ${\cal I}$, can be characterized by the orthogonal complements of the row spaces of $\mathbf{C}_k$'s, which is also known as the null spaces of $\mathbf{C}_k$'s. For $k=1, 2, \ldots, K$, let $V_k$ be the row space of $\mathbf{C}_k$. Denote the {\em orthogonal complement} of $V_k$ by $V_k^\perp$,
\[
V_k^\perp \triangleq \{ \mathbf{v} \in GF(q)^N:\, \mathbf{x} \cdot \mathbf{v} = 0 \text{ for all } \mathbf{x}\in V_k\},
\]
where $\mathbf{x} \cdot \mathbf{v}$ is the inner product of $\mathbf{x}$ and $\mathbf{v}$.
We will use the fact from linear algebra that a vector $\mathbf{x}$ is in $V_k$ if and only if  $\mathbf{x}\cdot\mathbf{v}=0$ for all $\mathbf{v} \in V_k^\perp$. Let $\mathbf{B}_k$ be an $(N - r_k) \times N$ matrix whose rows form a basis of $V_k^\bot$. To see whether a vector $\mathbf{x}$ is in $V_k$, it amounts to checking the condition $\mathbf{B}_k \mathbf{x} = \bzero$; if $\mathbf{B}_k \mathbf{x} = \bzero$, then $\mathbf{x}\in V_k$, and vice versa.

There are many different choices for the basis of the orthogonal complement $V_k^\perp$. We can obtain one such choice via the reduced row-echelon form (RREF) of $\mathbf{C}_k$. Suppose we have obtained the RREF of $\mathbf{C}_k$ by elementary row operations. By appropriately permutating the columns of $\mathbf{C}_k$, we can write $\mathbf{C}_k$ in the following form:
\begin{equation}
[ \mathbf{I}_{r_k} | \mathbf{A}_k] \mathbf{P}_k, \label{eq:RREF}
\end{equation}
where $\mathbf{I}_{r_k}$ is the $r_k \times r_k$ identity matrix, $\mathbf{A}_k$ is an $r_k \times (N-r_k)$ matrix over $GF(q)$, and $\mathbf{P}_k$ is an $N\times N$ permutation matrix\footnote{Recall that a permutation matrix is a square zero-one matrix so that each column and each row contain exactly one ``1''.}.  We can take
\begin{equation} \mathbf{B}_k = [-\mathbf{A}_k^T | \mathbf{I}_{N-r_k}]\mathbf{P}_k.
\label{eq:RREF2}
\end{equation}
The superscript $^T$ represents the transpose operator.
It is straightforward to verify that the product of the matrix in~\eqref{eq:RREF} and $\mathbf{B}_k^T$ is a zero matrix. Hence, the $n-r_k$ row vectors in $\mathbf{B}_k$ belong to the orthogonal complement $V_k^\perp$. Since $\mathbf{B}_k$ contains a permutation of $\mathbf{I}_{N-r_k}$ as one of its submatrices, the rows of $\mathbf{B}_k$ are linearly independent. As $\dim(V_k^\perp) = n-r_k$, we conclude that the rows of $\mathbf{B}_k$ form a basis of $V_k^\perp$.

In Appendix~\ref{app:B}, we give another way of computing a basis of $V_k^{\perp}$, which is suitable for incremental processing.


The following simple result characterizes the set of innovative encoding vectors, ${\cal I}$:

\begin{myle} \label{th:characterization}
Given $\mathbf{C}_1, \mathbf{C}_2, \ldots, \mathbf{C}_K$, an encoding vector $\mathbf{x}$ belongs to ${\cal I}$ if and only if $\mathbf{B}_k \mathbf{x} \neq \bzero$ for all $k$'s.
\end{myle}

\begin{proof}
If $\mathbf{B}_k \mathbf{x} \neq \bzero$, then $\mathbf{x}$ is not in $V_k$ and therefore, is innovative to user~$k$. It is innovative if $\mathbf{B}_k \mathbf{x} \neq \bzero$ for all $k$'s.

Conversely, if $\mathbf{B}_k \mathbf{x} = \bzero$ for some $k$, then $\mathbf{x}$ is in $V_k$, and hence is not innovative to user~$k$. Therefore, $\mathbf{x} \not \in {\cal I}$.
\end{proof}

When the underlying finite field size is small, innovative encoding vectors may not exist. For further investigation of the existence of innovative encoding vectors, we formulate the following decision problem:

{\bf Problem:} \IEVq

{\bf Instance:} $K$ matrices, $\mathbf{C}_1, \mathbf{C}_2,
 \ldots, \mathbf{C}_K$, over $GF(q)$, each of which has $N$ columns.

{\bf Question:} Is there an $N$-dimensional vector $\mathbf{x}$ over $GF(q)$ which does not belong to the row space of $\mathbf{C}_k$ for $k=1,2,\ldots, K$?

We can assume without loss of generality that all the matrices $\mathbf{C}_1$ to $\mathbf{C}_K$ in \IEVq\ are not full-rank. Also, we know from Theorem~\ref{nonemptyI} that the answer to \IEVq\ is always {\sc Yes} if $q \geq K$. The following result shows that the problem is NP-complete for $q=2$.

\begin{myth} \label{th:IEV}
\IEV\ is NP-complete.
\end{myth}

\begin{proof}
The idea is to Karp-reduce the 3-$\mathsf{SAT}$ problem, well-known to be NP-complete~\cite{GareyJohnson}, to the \IEV\ problem. Recall that the 3-$\mathsf{SAT}$ problem is a Boolean satisfiability problem, whose instance is a Boolean expression written in conjunctive normal form with three variables per clause (3-$\mathsf{CNF}$), and the question is to decide if there is some assignment of {\sc True} and {\sc False} vaules to the variables such that the given Boolean expression has a {\sc True} value.

Let $E$ be a given Boolean expression with $n$ variables $x_1,\ldots, x_n$, and $m$ clauses in 3-$\mathsf{CNF}$. We want to construct a Karp-reduction from the 3-$\mathsf{SAT}$ problem to the \IEV\ problem with $N=n+1$ packets and $K=m+1$ users.

For the $i$-th clause ($i=1, 2, \ldots, m$), we first construct a $3\times (n+1)$ matrix $\mathbf{B}_i$. If the $j$-th literal ($j=1,2,3$) in the $i$-th clause is $x_k$, then let the $k$-th component in the $j$-th row of $\mathbf{B}_i$ be one, and the other components be all zero. Otherwise, if the $j$-th literal in the $i$-th clause is $\neg x_k$, then let the $k$-th and the $(n+1)$-st component in the $j$-th row of $\mathbf{B}_i$ be both one, and the remaining components be all zero. Let $\mathbf{C}_i$ be a matrix whose rows form a basis of the orthogonal complement of the row space of $\mathbf{B}_i$.
We will use the fact that a vector $\mathbf{v}$ is in the row space of $\mathbf{C}_i$ if and only if $\mathbf{B}_i \mathbf{v} = 0$.

Consider an example with $n=4$ Boolean variables. From the clause $\neg x_1 \vee \neg x_2 \vee  x_3$, we get
\[\small \mathbf{B}_i =
\begin{bmatrix}
1& 0& 0& 0& 1 \\
0& 1& 0& 0& 1 \\
0& 0& 1& 0& 0
\end{bmatrix}, \
 \mathbf{C}_i = \begin{bmatrix}
 0 & 0 & 0 & 1 & 0 \\
 1 & 1 & 0 & 0 & 1
 \end{bmatrix}.
\]
It can be verified that each row in $\mathbf{B}_i$ is orthogonal to the rows in $\mathbf{C}_i$, i.e.,
the row space of $\mathbf{C}_i$ is the orthogonal complement of the row space of $\mathbf{B}_i$.

For the extra user, user $m+1$, let $\mathbf{B}_{m+1}$ be the $1\times (n+1)$ matrix $[\mathbf{0}_n \; 1]$, where $\mathbf{0}_n$ stands for the $1\times n$  all-zero vector. The problem reduction can be done in polynomial time.

Let $\mathbf{x} = (x_1, x_2, \ldots, x_n)$ be a Boolean vector and define $\hat{\mathbf{x}} \triangleq (\mathbf{x}, 1)$. Note that any solution $\mathbf{x}$ to a given 3-$\mathsf{SAT}$ problem instance would cause the product $\mathbf{B}_j \hat {\mathbf{x}}$ a non-zero vector for $j=1, 2, \ldots, m+1$. By Lemma~\ref{th:characterization}, $\hat{\mathbf{x}}$ is not in the row space of $\mathbf{C}_j$ for all $j$. Hence $\hat{\mathbf{x}}$ is also a solution to the derived \IEV\  problem.


Conversely, any solution to the derived \IEV\ problem also yields a solution to the original 3-$\mathsf{SAT}$  problem as well. Let $\mathbf{c} = (c_1, c_2, \ldots, c_n, c_{n+1}) \in GF(2)^{n+1}$ be a solution to the derived \IEV\ problem. Note that we must have $c_{n+1} = 1$ because of $\mathbf{B}_{m+1}$. Let  $i$ be an index between 1 and~$m$. Since $\mathbf{c}$ is not in the row space of $\mathbf{C}_i$, the product $\mathbf{B}_i \mathbf{c}$ is a non-zero vector. Hence, if we assign {\sc True} to $x_k$ if $c_k=1$ and {\sc False} to $x_k$ if $c_k=0$, for $k=1,2,\ldots,n$, then the $i$-th clause will have a {\sc True} value. Since this is true for all $i$, the whole Boolean expression also has a {\sc True} value.

The problem \IEV\ is clearly in NP, since it is efficiently verifiable. Hence it is NP-complete.
\end{proof}

The above proof can be extended to the following more general result:

\begin{myth} For any prime power $q$, the problem
\IEVq\ is NP-complete.
\end{myth}

\begin{proof}
The reduction from 3-$\mathsf{SAT}$ to \IEVq\ is the same as before, except with the following changes:
\begin{enumerate}
\item In the derived \IEVq\ problem, $K = m + 1 + n(q-2)$ users. In other words, there are $n(q-2)$ more users than that in the previous proof.
\item For $i=1, 2, \ldots, m$, $\mathbf{B}_i$ is defined the same as before except that when the $j$-th literal in the $i$-th clause is $\neg x_k$, let the $(n+1)$-th component be $-1$ (rather than 1).
\item For $u = 1, 2, \ldots, n$, let $\mathbf{E}_{u,1}, \mathbf{E}_{u,2}, \ldots, \mathbf{E}_{u,q-2}$ be $1 \times (n+1)$ matrices whose $u$-th components are distinct elements in $\text{GF}(q)\setminus\{0,1\}$ and the ($n$+1)-st components are all equal to $-1$. Let these $n(q-2)$ matrices be $\mathbf{B}_i$'s for $i = m+2, m+3, \ldots, m+1+n(q-2)$.
\end{enumerate}
The forward part is the same as before, so we only need to consider the converse part. Let $\mathbf{c} = (c_1, c_2, \ldots, c_n, c_{n+1}) \in GF(q)^{n+1}$ be a solution to the derived \IEVq\ problem. Same as before, we must have $c_{n+1} \neq 0$ because of $\mathbf{B}_{m+1}$. Since a non-zero scalar multiplication of $\mathbf{c}$ remains to be a solution to the derived \IEVq\ problem, without loss of generality, we can assume that $c_{n+1} = 1$. Due to the extra $n(q-2)$ users, for $i=1,2, \ldots, n$, we must have $c_i = 0$ or 1, for otherwise $\mathbf{E}_{i,j} \mathbf{c}$ must be zero for some $j$. (More precisely, one and only one of these $q-2$ vectors is zero.) The rest of the proof then follows the same argument as in Theorem~\ref{th:IEV}.
\end{proof}

Note that in the formulation of \IEVq\, the values of $N$ and $K$ are arbitrary. The above result shows that it is NP-complete. With the restriction of $K \leq q$, \IEVq\ becomes trivial to solve, as shown by Theorem~\ref{nonemptyI}.


Apart from the problem of existence of an innovative vector, it is also of interest in finding an $N$-dimensional encoding vector that is innovative to as many users as possible. We state the optimization problem as follows:

{\bf Problem:} \MAXIEVq

{\bf Instance:} $K$ matrices $\mathbf{C}_k$ over $GF(q)$, $k=1,2,\ldots, K$, and each matrix has $N$ columns.

{\bf Objective:} Find an $N$-dimensional vector $\mathbf{x}$ over $GF(q)$ such that the number of users to whom $\mathbf{x}$ is innovative is maximized.


The following result shows the hardness of finding approximate solution to \MAXIEVq:

\begin{myth}
For $2 \leq q < K$, there is no approximation algorithm for \MAXIEVq\ with an approximation guarantee of $1 - \epsilon_M$, assuming $\text{P} \neq \text{NP}$, where $\epsilon_M$ is a positive constant.
\end{myth}

\begin{proof}
Given a Boolean expression in the 3-$\mathsf{CNF}$ form, the problem of maximizing the number of clauses
that have TRUE values is commonly called the $\mathsf{Max}$-3-$\mathsf{SAT}$ problem. Consider the same reduction
described in the proof of Theorem~\ref{th:IEV}. It is clear that the number of clauses that have {\sc True} values
under a given Boolean vector $\mathbf{x}$ is the same as the number of users to whom $\hat{\mathbf{x}} = (\mathbf{x}, 1)$ is innovative, excluding user~$m+1$. Therefore, the reduction is a gap-preserving reduction from $\mathsf{Max}$-3-$\mathsf{SAT}$ to \MAXIEVq. The statement then follows from \cite[Corollary 29.8]{vazirani}.
\end{proof}

%

\section{The Sparsity Problem} \label{sec:sparsity}

Decoding complexity is one of the critical issues that could determine the practicality of linear network coding in broadcast erasure channels. One way to reduce the decoding complexity is to generate sparse encoding vectors and apply a decoding algorithm that exploits the sparsity of encoding vectors at receivers. In this section, we focus on the sparsity issues of innovative encoding vectors.

\subsection{Existence of $K$-sparse innovative vector} \label{sec:ksparse}

In the previous section, it is found that innovative vectors always exist if $q\geq K$. In fact, we can prove a stronger statement that $K$-sparse innovative vectors always exist under the same condition.

\begin{myle} \label{le:key}
For $k=1,2,\ldots, K$, let $f_k(\mathbf{x})$ be a non-zero linear polynomials in $L$ variables
$$
f_k(\mathbf{x}) \triangleq \alpha_{k1} x_1 + \alpha_{k2} x_2 + \cdots + \alpha_{kL} x_L,  \; k = 1, 2, \ldots, K,
$$
where the coefficients are elements in $GF(q)$. If $q \geq K$, we can always find a vector $\mathbf{x}^* = (x_1, x_2, \ldots, x_L) \in GF(q)^L$ such that $f_k(\mathbf{x}^*) \neq 0$ for all~$k$.
\end{myle}

We first give a combinatorial proof:

\begin{proof}
For $k=1,2,\ldots, K$, let $V_k$ be the set of vectors $\mathbf{x}$ in $GF(q)^L$ satisfying $f_k(\mathbf{x}) = 0$. The set $V_k$ is a subspace of dimension $L-1$. By Lemma~\ref{special_case_of_Schwartz}, the cardinality of the union of these $K$ subspaces is strictly less than $K q^{L-1}$ elements, which in turn is less than or equal to the cardinality of the whole space $GF(q)^L$. Thus there exists at least one vector $\mathbf{x}^*$ in $GF(q)^L$ such that $f_k(\mathbf{x}^*) \neq 0$ for all~$k$.
\end{proof}

Now we give an alternative proof of Lemma~\ref{le:key}, which is algorithmic and constructive:

\begin{proof}
Let ${\cal S}_l$, where $l=1,2,\ldots,L$, be the index set such that $k \in {\cal S}_l$ if and only if $\alpha_{kl} \neq 0$. Since none of the linear polynomials $f_k(\mathbf{x})$'s are identically zero, the union $\bigcup_{l=1}^L \mathcal{S}_l$ is equal to $\{1,2,\ldots, K\}$.
We distinguish two cases:

{\em Case 1:} $|{\cal S}_l| = K$ for some $l$. We can simply let $x_l^* = 1$ and $x_n^* = 0$ for $n\neq l$.

{\em Case 2:} $|{\cal S}_l| < K$ for all $l$. We assign values to the variables iteratively. Suppose we have already assign
$x_1^*, x_2^*, \ldots, x_{t-1}^*$ to the first $t-1$ variables. We note that  $f_k(x_1^*,\ldots, x_{t-1}^*, x_t, 0, \ldots, 0)=0$ is a linear equation in a single variable $x_t$, and thus have only one solution.  As $|\mathcal{S}_l| < K \leq q$, the number of elements in $GF(q)$ which satisfy $f_k(x_1^*, \ldots, x_{t-1}^*, x_t, 0, \ldots, 0)=0$ for some $k\in \mathcal{S}_l$ is strictly less than $q$. There must exist $x_t^* \in GF(q)$ such that
\[
 f_k(x_1^*, x_2^*,\ldots, x_{t-1}^*, x_t^*, 0, 0, \ldots, 0) \neq 0
\]
for all $k\in \mathcal{S}_l$. Upon termination, it is guaranteed that $f_k(x_1^*, x_2^*,\ldots, x_L^*) \neq 0$ for all $k$.
\end{proof}

We call the method in the proof of Lemma~\ref{le:key} the {\em Sequential Assignment} (SA) algorithm. Its computational complexity in terms of number of multiplications/divisions over $GF(q)$ is analyzed as follows: In this algorithm, there are $L$ iterations. In each iteration, we need to find an element in $GF(q)$ that is not a root of any of these $K$ equations. Consider the $t$-th iteration. For the $k$-th equation, we need to compute $\alpha_{k,t-1} x_{t-1}^*$ and add it to the accumulated sum $\sum_{j=1}^{t-2} \alpha_{kj} x_j$, which is stored for the next iteration. The root of this equation can then be obtained by a division. Therefore, the total complexity of SA is $O(KL)$.


\begin{myexam}
Let
\begin{align*}
f_1(\mathbf{x}) &\triangleq x_1 + 2x_2\\
f_2(\mathbf{x}) &\triangleq x_2 + 2x_3\\
f_3(\mathbf{x}) &\triangleq 2x_1 + x_3
\end{align*}
be $K=3$ linear polynomials over $GF(3)$. We apply the SA algorithm to find an assignment of $\mathbf{x}=(x_1, x_2, x_3)$ such that $f_1(\mathbf{x})$, $f_2(\mathbf{x})$ and $f_3(\mathbf{x})$ are all non-zero. First of all, the three index sets are
$\mathcal{S}_1 = \{1,3\}$, $\mathcal{S}_2 = \{1,2\}$, and
$\mathcal{S}_3 = \{2,3\}$.
None of them has cardinality three. We proceed as described in the second case. We assign an arbitrary non-zero value to $x_1$, say $x_1=1$, and we can check that
$f_1(1,0,0) = 1$, $f_2(1,0,0) = 0$, $f_3(1,0,0)=2$.

Next, we want to find $x_2 \in GF(3)$ such that
\begin{align*}
f_1(1,x_2,0) &= 1 + 2x_2 \neq 0, \text{ and}\\
f_2(1,x_2,0) &= x_2 \neq 0.
\end{align*}
It turns out that the only choice for $x_2$ is $x_2=2$. After $x_2$ is fixed, we search for $x_3 \in GF(3)$ such that
\begin{align*}
f_2(1,2,x_3) &= 2+2x_3 \neq 0 \\
f_3(1,2,x_3) &= 2+x_3 \neq 0.
\end{align*}
The only choice for $x_3$ is $x_3=0$. Finally, we check the values of $f_1$, $f_2$ and $f_3$ evaluated at $\mathbf{x}=(1, 2, 0)$ as follows:
\[
f_1(1,2,0) = f_2(1,2,0) = f_3(1,2,0) = 2 \neq 0.
\]
\begin{flushright}
$\Box$
\end{flushright}
\end{myexam}

Lemma~\ref{le:key} can be used to establish the following result:

\begin{myth} \label{th:sparse}
If $q \geq K$, there exists a $K$-sparse encoding vector in ${\cal I}$.
\end{myth}

\begin{proof}
For $k=1, 2, \ldots, K$, let $\mathbf{b}_k^T$ be an arbitrary row vector in $\mathbf{B}_k$, and let $n_k$ be an arbitrary index such that the $n_k$-th component of $\mathbf{b}_k$ is non-zero.
Form a new index set ${\cal N}$ that consists of all $n_k$'s. The cardinality of ${\cal N}$ may be less than $K$ since the $n_k$'s may not be distinct. Let $\mathbf{b}_k({\cal N})$ be a truncated vector of $\mathbf{b}_k$, which consists of only the components of $\mathbf{b}_k$ whose indices are in~${\cal N}$. Its dimension is equal to $|{\cal N}| \leq K$.

Now we show that there exists a vector $\mathbf{x} \in {\cal I}$ such that the $i$-th component of $\mathbf{x}$ is equal to zero if $i \not \in {\cal N}$.
If the $i$-th component of $\mathbf{x}$ is zero for all $i\not\in \mathcal{N}$, then the inner product of $\mathbf{b}_k$ and $\mathbf{x}$ is the same as the inner product of $\mathbf{b}_k(\mathcal{N})$ and $\mathbf{x}(\mathcal{N})$. According to Lemma~\ref{th:characterization}, $\mathbf{x}$ is in ${\cal I}$ if $\mathbf{b}_k({\cal N}) \cdot \mathbf{x}({\cal N}) \neq 0$ for all $k$'s. By Lemma~\ref{le:key}, we can find such a vector $\mathbf{x}$ if $q \geq K$. Clearly, such a vector has $\mathcal{N}$ as its support, and is hence $K$-sparse.
\end{proof}

The above result shows that if $q\geq K$, the minimum Hamming weight of innovative vectors is bounded above by $K$. This upper bound cannot be further reduced as the following example shows:

\begin{myexam}
Consider a broadcast system of $K$ users and $N$ packets, where $N \geq K$. Suppose that user $k$ has received a set of uncoded packets $\mathcal{A}_k$. Here we regard $\mathcal{A}_k$ as a subset of $\{1,2,\ldots, N\}$.
Furthermore, suppose that the complement of the $\mathcal{A}_k$'s are mutually disjoint, i.e., $\mathcal{A}_j^c \cap \mathcal{A}_k^c = \emptyset$ for $j\neq k$.
In such a scenario, an innovative packet must be a linear combination of at least $K$ packets. For example, let $N=4$ and $K=3$. If the encoding matrices of the three users are
\[
\begin{bmatrix}
1 & 0& 0& 0 \\
0 & 1& 0& 0
\end{bmatrix}, \
\begin{bmatrix}
1 & 0& 0& 0 \\
0 & 0& 1& 0 \\
0 & 0& 0& 1
\end{bmatrix}, \
\begin{bmatrix}
0 & 1& 0& 0 \\
0 & 0& 1& 0 \\
0 & 0& 0& 1
\end{bmatrix},
\]
then an innovative encoding vector must have Hamming weight at least 3. For instance $(1, 1, 1 ,0)$ and $(1,1,0,1)$ are innovative, but no vector with Hamming weight 2 or less is innovative. \hfill $\Box$
\end{myexam}

\subsection{Sparsest Innovative Vectors}

Theorem~\ref{th:sparse} shows that we can always find a $K$-sparse innovative vector if $q \geq K$. It serves as an upper bound on the minimum Hamming weight of innovative vectors. To further reduce the decoding complexity, it is natural to consider the issue of finding the sparsest innovative encoding vector for given $\mathbf{C}_k$'s. In other words, we want to find a vector in ${\cal I}$ that has the minimum Hamming weight for the case where $q \geq K$. We call this algorithmic problem {\sc Sparsity}. We state its decision version formally as follows:

\smallskip

{\bf Problem:} {\sc Sparsity}

{\bf Instance:} A positive integer $n$ and $K$ matrices with $N$ columns, $\mathbf{C}_1, \mathbf{C}_2, \ldots, \mathbf{C}_K$, over $GF(q)$, where $q \geq K$.

{\bf Question:} Is there a vector $\mathbf{x} \in \mathcal{I}$ with Hamming weight less than or equal to $n$?

We have already proven that the answer is always {\sc Yes} if $n \geq K$. We are interested in the case where $n < K$.

Given all $\mathbf{C}_k$'s, we can find a basis  $\mathbf{B}_k$'s  of their corresponding null spaces by the method mentioned in Section~\ref{sec:secinno}. For $k=1, 2, \ldots, K$, let $\mathbf{b}_{k,i}^T$ be the $i$-th row of $\mathbf{B}_k$. We define
\begin{equation}
\tilde{\mathbf{b}}_k \triangleq \vee_{i=1}^{N-r_k} \mathbf{b}_{k,i}, \label{btilde}
\end{equation}
where $\vee$ denotes the logical-OR operator applied component-wise to the $N-r_k$ vectors, with each non-zero component being regarded as a ``1''. In other words, the $j$-th component of $\tilde{\mathbf{b}}_k$ is one if and only if the $j$-th column of $\mathbf{B}_k$ is nonzero. We define $\mathbf{B}$ as the $K\times N$ matrix whose $k$-th row is equal to $\tilde{\mathbf{b}}_k^T$. Note that $\mathbf{B}$ is a binary matrix and has no zero rows. For a matrix $\mathbf{A}$ and a subset $\cal N$ of the column indices of $\mathbf{A}$, let $\mathbf{A}(\cal N)$ be the $K \times |{\cal N}|$ submatrix of $\mathbf{A}$, whose columns are chosen according to ${\cal N}$. We need the following lemma:

\begin{myle} \label{le:index_set}
Let ${\cal N} \subseteq \{1, 2, \ldots, N\}$ be an index set and $q \geq K$. There exists an encoding vector $\mathbf{x} = (x_1, x_2, \ldots, x_N) \in {\cal I}$ over $GF(q)$ with $supp(\mathbf{x}) \subseteq {\cal N}$ if and only if $\mathbf{B}({\cal N})$ has no zero rows.
\end{myle}

\begin{proof}
If $\mathbf{B}({\cal N})$ has no zero rows, then $\tilde{\mathbf{b}}_k({\cal N}) \neq \bzero$ for all $k$'s. Furthermore, for all $k$'s, there must exist $\mathbf{b}_{k,j}({\cal N}) \neq \bzero$ for some $j$.  By Lemma~\ref{le:key}, we can find $\mathbf{x}({\cal N}) \in GF(q)^{|{\cal N}|}$ such that $\mathbf{b}_{k,j}({\cal N}) \cdot \mathbf{x}({\cal N}) \neq 0$ for all $k$'s. Let the components of $\mathbf{x}$ whose indices do not belong to ${\cal N}$ be zero. Then by Lemma~\ref{th:characterization}, $\mathbf{x} \in {\cal I}$.

Conversely, if $\mathbf{x}$ is an innovative vector with $x_n = 0$ for $n \not \in {\cal N}$, then $\mathbf{B}({\cal N})$ cannot have zero rows, for if row $k$ of $\mathbf{B}({\cal N})$ is a zero vector, then $\mathbf{B}_k({\cal N})$ is a zero matrix and the $k$-th inequality in Lemma~\ref{th:characterization} cannot hold.
\end{proof}

The NP-completeness of {\sc Sparsity} can be established by reducing the hitting set problem, {\sc HittingSet}, to {\sc Sparsity}. Recall that a problem instance of {\sc HittingSet} consists of a collection $\mathscr{C}$ of subsets of a finite set ${\cal U}$. A {\em hitting set} for $\mathscr{C}$ is a subset of ${\cal U}$ such that it contains at least one element from each subset in $\mathscr{C}$. The decision version of this problem is to determine whether we can find a hitting set with cardinality less than or equal to a given value.

{\bf Problem:} {\sc HittingSet}

{\bf Instance:} A finite set $\mathcal{U}$, a collection $\mathscr{C}$ of subsets of $\mathcal{U}$ and an integer $n$.

{\bf Question:} Is there a subset $\mathcal{S} \subseteq \mathcal{U}$ with cardinality less than or equal to $n$ such that for each $\mathcal{C} \in \mathscr{C}$ we have $\mathcal{C} \cap \mathcal{S} \neq \emptyset$?

It is well known that {\sc HittingSet} is NP-complete~\cite{GareyJohnson}.

\begin{myexam}
Let $\mathcal{U} = \{1,2,3,4,5\}$,
\[\mathscr{C} = \{ \{1,2,3\}, \{2,3,4\}, \{4,5\}\}
\]
and $n=2$. We can check that $\{1,4\}$ is a hitting set of size $n=2$. \hfill $\Box$
\end{myexam}

\begin{myth} \label{th:NP}
{\sc Sparsity} is NP-complete.
\end{myth}

\begin{proof}
We are going to reduce {\sc HittingSet} to an instance of {\sc Sparsity} via a Karp-reduction.
Let the cardinality of ${\cal U}$ be $N$. Label the elements of ${\cal U}$ by $1, 2, \ldots, N$. We define $\mathscr{C} \triangleq \{{\cal{C}}_1, {\cal{C}}_2, \ldots, {\cal{C}}_K\}$, where $K$ is the number of non-empty subsets in $\mathscr{C}$. For $k=1, 2, \ldots, K$, form an $N$-vector $\mathbf{b}_k \in GF(q)^{N}$ with its $i$-th component equal to one if $i$ is in ${\cal{C}}_k$ and zero otherwise, i.e.,  $\mathbf{b}_k$ is the characteristic vector of  ${\cal{C}}_k$. Note that $\mathbf{b}_k \neq \bzero$ and $\mathscr{C} = \{supp(\mathbf{b}_1), supp(\mathbf{b}_2), \ldots, supp(\mathbf{b}_K)\}$. These $\mathbf{b}_k$'s correspond to the degenerate form of $\mathbf{B}_k$'s in Lemma~\ref{th:characterization} with only one row in $\mathbf{B}_k$. Let $\mathbf{C}_k$ be the encoding matrix of user~$k$, whose row space is the null space of $\mathbf{B}_k$ and ${\cal I}$ be the innovative vector set defined in (\ref{inno}). In other words, any instance of {\sc HittingSet} can be represented as an instance of {\sc Sparsity} in polynomial time.

It remains to show that there exists a hitting set $\cal{H}$ for $\mathscr{C}$ with $|{\cal{H}}| \leq n$ if and only if there exists an $\mathbf{x} \in {\cal I}$ with Hamming weight $|supp(\mathbf{x})| \leq n$. Given the $\mathbf{b}_k$'s obtained via the above reduction, suppose there exists $\mathbf{x} \in {\cal I}$ with $|supp(\mathbf{x})| \leq n$. By Lemma~\ref{th:characterization}, we must have $\mathbf{b}_k \cdot \mathbf{x} \neq 0$ for all $k$'s, which implies $supp(\mathbf{b}_k) \cap supp(\mathbf{x}) \neq \emptyset$ for all $k$'s. The set $supp(\mathbf{x})$ is therefore a hitting set for the given instance. Conversely, given a hitting set $\cal{H}$ for $\mathscr{C}$ with $|{\cal{H}}| \leq n$, by definition $supp(\mathbf{b}_k) \cap {\cal H} \neq \emptyset$ for all $k$'s. Therefore, $\mathbf{B}({\cal H})$ has no zero rows. By Lemma~\ref{le:index_set}, there exists an
$\mathbf{x} \in  GF(q)^{N}$ such that $supp(\mathbf{x}) \subseteq {\cal H}$. Hence, $|supp(\mathbf{x})| \leq n$.

As {\sc Sparsity} is verifiable in polynomial time, {\sc Sparsity} is in NP. Hence it is NP-complete.
\end{proof}

Now we define the optimization version of {\sc Sparsity} as follows:

{\bf Problem:} {\sc Max Sparsity}

{\bf Instance:} A positive integer $n$ and $K$ matrices with $N$ columns, $\mathbf{C}_1, \mathbf{C}_2, \ldots, \mathbf{C}_K$, over $GF(q)$, where $q \geq K$.

{\bf Objective:} Find a vector $\mathbf{x} \in \mathcal{I}$ with minimum Hamming weight.

We call the minimum Hamming weight among all innovative vectors the {\em sparsity number}, and denote it by~$\omega$. It is easy to see that if a polynomial-time algorithm can be found for solving the optimization version of {\sc Sparsity}, then that algorithm can be used for solving the decision version of {\sc Sparsity} in polynomial time as well. Therefore, {\sc Max Sparsity} is NP-hard.

On the other hand, if $K$ is held fixed, meaning that the problem size grows only with $N$, then there exists algorithm whose complexity grows polynomially in $N$ to solve {\sc Max Sparsity}. It is proven in \cite{KSS11} and Section~\ref{sec:ksparse} that a $K$-sparse vector exists in ${\cal I}$, if $q\geq K$. By listing all vectors in $GF(q)^N$ with Hamming weight less than or equal to $K$, we can use Lemma~\ref{th:characterization} to check whether each of them is in ${\cal I}$. For each $K$-sparse encoding vector, we compute the matrix product $\mathbf{B}_k \mathbf{x}$ for $k=1,2,\ldots, K$. Each matrix product takes $O(NK)$ finite field operations. The total number of finite field operations for each candidate $\mathbf{x}$ is $O(NK^2)$. After checking all $K$-sparse encoding vectors, we can then find one with minimum Hamming weight. The number of non-zero vectors in $GF(q)^N$ with Hamming weight no more than $K$ is equal to $\sum_{k=1}^K \binom{N}{k} (q-1)^{k}$. For fixed $K$ and $q$, the summation is dominated by the largest term $\binom{N}{K}(q-1)^K$ when $N$ is large, which is of order $O(N^K)$. The brute-force method can solve the problem with time complexity of $O(N^K (NK^2))$. As $K$ is held fixed, {\sc Max Sparsity} can be solved in polynomial time in~$N$.

Let {\sc Min HittingSet} be the minimization version of the hitting set problem, in which we want to find a hitting set with minimum cardinality.
The next result shows that {\sc Max Sparsity} can be solved via {\sc Min HittingSet} based on the concept of Levin-reduction.

\begin{myth} \label{th:levin_reduction}
{\sc Max Sparsity} can be Levin-reduced to {\sc Min HittingSet}.
\end{myth}

\begin{proof}
Given an instance of {\sc Max Sparsity}, we determine $\tilde{\mathbf{b}}_k$ as in \eqref{btilde} for $k = 1, 2, \ldots, K$. Then we form the following instance of {\sc Min HittingSet}:
\begin{align*}
\mathcal{U} &= \{1, 2, \ldots, N\}, \\
\mathscr{C} &= \{ supp(\tilde{\mathbf{b}}_1), supp(\tilde{\mathbf{b}}_2), \ldots, supp(\tilde{\mathbf{b}}_K) \}.
\end{align*}
Let $\mathcal{H}$ be a solution to the above instance. Then $\mathbf{B}(\mathcal{H})$ has no zero rows. By Lemma~\ref{le:index_set}, there exists
a vector $\mathbf{x}^* \in \mathcal{I}$ over $GF(q)$ with $supp(\mathbf{x}^*) \subseteq \mathcal{H}$. Such a vector $\mathbf{x}^*$ can be found by the SA algorithm in polynomial time.

We claim that there does not exist $\mathbf{x}' \in \mathcal{I}$ with Hamming weight $|supp(\mathbf{x}')| < |\mathcal{H}|$, and thus $|supp(\mathbf{x}^*)|$ must equal $|\mathcal{H}|$. Suppose there exists such a vector $\mathbf{x}'$. Lemma~\ref{le:index_set} implies that $\mathbf{B}(supp(\mathbf{x}'))$ has no zero rows, which in turn implies that $supp(\mathbf{x}') \cap supp(\tilde{\mathbf{b}}_k) \neq \emptyset$ for all $k$'s. Then $supp(\mathbf{x}')$ would be a hitting set with cardinality strictly less than $|\mathcal{H}|$. A contradiction.

The proof is completed by matching the relevant entities and procedures with those in Definition~\ref{def:levin}. Note that the transformation of a given instance of {\sc Max Sparsity} to an instance of {\sc Min HittingSet} in essence corresponds to the mapping $f$. A solution to an instance of {\sc Min HittingSet}, $\mathcal{H}$, corresponds to $y'$. Obtaining $\bx^*$ from $\mathcal{H}$ by the SA algorithm corresponds to the mapping $g$.
\end{proof}



\section{Network Coding Algorithms} \label{sec:IVGA}

In this section, we present algorithms that generate sparse innovative encoding vectors for $q \geq K$. While for the binary case (i.e., $q = 2$), finding an innovative encoding vector may not always be possible, a modification of the algorithm is also proposed for handling it.

\subsection{The Optimal Hitting Method}


For $q \geq K$, we generate a sparest innovative vector in two steps. First we find an index set ${\cal N}$ with minimum cardinality, which determines the support of the innovative encoding vector.  This is accomplished by solving the hitting set problem. Once ${\cal N}$ is found, the non-zero entries in the vector can be obtained by the SA algorithm.

The hitting set problem can be solved exactly by binary integer programming (BIP), formulated as follows:
$$
\omega = \min_{\mathbf{y}} y_1+y_2+\ldots+y_N, \label{thm:BIP}
$$
subject to
\begin{align*}
\mathbf{B} \mathbf{y} &\geq \bone,
\end{align*}
where $$\mathbf{B} = \begin{bmatrix}\tilde{\mathbf{b}}_{1}\\\tilde{\mathbf{b}}_{2}\\ \vdots \\\tilde{\mathbf{b}}_{K}
\end{bmatrix}$$
is a $K\times N$ binary matrix,  $\mathbf{y} = (y_1,y_2, \ldots, y_N)$ is an $N$-dimensional binary vector, the vector $\bone$ is the $K$-dimensional all-one vector, and the inequality sign is applied component-wise.


To solve the above problem, we can apply any algorithm for solving BIP in general, for example the cutting plane method. We refer the readers to~\cite{BIP} for more details on BIP.

\begin{myexam}
Let $q=3$, $K=3$ and $N=4$, and the orthogonal complements of $V_1$, $V_2$ and $V_3$ be given respectively by the row spaces of
\begin{gather*}
\mathbf{B}_1 = \begin{bmatrix} 1&2&0&1 \\ 1&1&0&0 \end{bmatrix},\; \mathbf{B}_2 = \begin{bmatrix} 0&2&1&0 \end{bmatrix}, \\
\mathbf{B}_3 = \begin{bmatrix} 0&0&1&1 \\ 1&0&0&2 \end{bmatrix}.
\end{gather*}
The vectors $\tilde{\mathbf{b}}_k$, for $k=1,2,3$, are
\begin{gather*}
\tilde{\mathbf{b}}_1 = [ 1\ 1\ 0\ 1 ],\; \tilde{\mathbf{b}}_2 = [ 0\ 1\ 1\ 0 ],\; \tilde{\mathbf{b}}_3 = [ 1\ 0\ 1\ 1 ].
\end{gather*}
The corresponding instance of {\sc Min HittingSet} is:
\begin{align*}
\mathcal{U} = \{1,2,3,4\},\; \mathscr{C} = \{ \{1,2,4\}, \{2,3\}, \{1,3,4\} \}.
\end{align*}
The solution to both {\sc Max Sparsity} and {\sc Min HittingSet} can be obtained by solving the following BIP:
\[
 \min y_1 + y_2 + y_3 + y_4,
\]
subject to
\begin{gather*}
y_1+y_2+y_4 \geq 1, \; y_2+y_3 \geq 1, \; y_1+y_3+y_4 \geq 1, \\
y_1, y_2, y_3, y_4 \in \{0,1\}.
\end{gather*}

One optimal solution is $y_1 = y_2 = 1$ and $y_3 = y_4 = 0$. That means, the sparsity number, $\omega$, is equal to two and ${\cal N} = \{1,2\}$. Furthermore, according to Lemma~\ref{le:index_set}, a 2-sparse innovative encoding vector can be found, for example, by the SA algorithm. \hfill $\Box$
\end{myexam}


We call the above procedure for generating an innovative vector with minimum Hamming weight the {\em Optimal Hitting (OH) method}. We summarize the algorithm as follows:

\begin{flushleft}
{\bf The Optimal Hitting method (OH):}
\end{flushleft}

\noindent {\bf Input:} For $k=1,2,\ldots , K$,  full-rank $r_k\times N$ matrix $\mathbf{C}_k$ over $GF(q)$, where $q \geq K$ and $0\leq r_k < N$.

\noindent {\bf Output:} $\mathbf{x} = (x_1, x_2, \ldots, x_N) \in \mathcal{I}$ with minimum Hamming weight.

\smallskip

\noindent {\bf Step 0:}  Initialize $\mathbf{x}$ as the zero vector.

\noindent {\bf Step 1:} For  $k=1,2,\ldots, K$, obtain a basis of the null space of $\mathbf{C}_k$. Let $\mathbf{B}_k$ be the $(N-r_k)\times N$ matrix over $GF(q)$ whose $j$-th row is the $j$-th vector in the basis.

\noindent {\bf Step 2:} For  $k=1,2,\ldots, K$,
let $\tilde{\mathbf{b}}_k$ be the component-wise logical-OR operations to the $N-r_k$ row vectors of $\mathbf{B}_k$. (Each non-zero component of $\mathbf{B}_k$ is regarded as ``1'' when taking the logical-OR operation.)

\noindent  {\bf Step 3:} Solve the corresponding {\sc Min HittingSet} as shown in Theorem~\ref{th:levin_reduction} and return $\mathcal{H}$.


\noindent {\bf Step 4:} For $k=1,2,\ldots, K$, choose a row vector from $\mathbf{B}_k$, say $\hat{\mathbf{b}}_k^T$,  such that $supp(\hat{\mathbf{b}}_k)\cap \mathcal{H} \neq \emptyset$.

\noindent  {\bf Step 5:} Determine $\mathbf{x}(\mathcal{H})$ such that $\mathbf{x}(\mathcal{H}) \cdot \hat{\mathbf{b}}_k(\mathcal{H}) \neq \bzero$ for $k = 1, 2, \ldots, K$, by the SA algorithm.

%
%

\medskip

\noindent {\bf Example 4 (continued).}
We solve the hitting set problem in Step 3 and obtain $\mathbf{y} = (1,1,0,0)$. Hence, $\mathcal{N} =\{1,2\}$. In Step~4, we choose
$$\hat{\mathbf{b}}_1 = (1,2,0,1), \hat{\mathbf{b}}_2 = (0,2,1,0) \text{ and }\hat{\mathbf{b}}_3 = (1,0,0,2).$$
In Step~5, we obtain $\mathcal{S}_{1} = \{1,3\}$ and $\mathcal{S}_{2} = \{1,2\}$. Note that both $|\mathcal{S}_{1}|$ and $|\mathcal{S}_{2}|$ are not equal to~3. We next set  $x_1 = 1$, and choose $x_2$ such that
\begin{align*}
\hat{\mathbf{b}}_{1} \cdot (1, x_2, 0, 0) &\neq 0 \\
\hat{\mathbf{b}}_{2} \cdot (1, x_2, 0, 0) &\neq 0 .
\end{align*}
We can choose $x_2 = 2$ to satisfy these two inequalities simultaneously. The vector $\mathbf{x} = (1,2,0,0)$ is an innovative encoding vector of minimum Hamming weight. \hfill $\Box$

\subsection{The Greedy Hitting Method}

Step~3 in the OH method requires solving an NP-hard problem. Therefore, some computationally efficient heuristics should be considered in practice. It is well known that {\sc Min HittingSet} can be solved approximately by the following greedy approach~\cite{Setcover}:
\begin{itemize}
\item Repeat until all sets of $\mathscr{C}$ are hit:
    \begin{itemize}
    \item Pick the element that hits the largest number of sets that have not been hit yet.
    \end{itemize}
\end{itemize}
In Step 3 of the OH method, the above greedy algorithm can be used to find approximate solutions. We call this modification the {\em Greedy Hitting} (GH) method.

\begin{myth}
The GH method is an $H_N$ factor approximation algorithm for {\sc Max Sparsity}, where $H_\ell$ is the $\ell$-th harmonic number, defined as $H(\ell) \triangleq \sum_{k=1}^\ell \frac{1}{k}$.
\end{myth}

\begin{proof}
It is well known that the hitting set problem is just a reformulation of the set covering problem. Therefore, the greedy algorithm is an $H_{|\mathcal{U}|}$ factor approximation algorithm for {\sc Min HittingSet}, as well as for the set covering problem \cite{vazirani}. As shown in Theorem~\ref{th:levin_reduction}, {\sc Max Sparsity} can be reduced to {\sc Min HittingSet}, and the sparsity number is equal to the cardinality of the minimum hitting set. Hence, GH is also an $H_N$ factor approximation algorithm for {\sc Max Sparsity}.
\end{proof}




Now we analyze the computational complexity of the proposed OH and GH methods. For the OH method, the computation of each $\mathbf{B}_k$ can be reduced to the computation of the RREF of~$\mathbf{C}_k$, which takes $O(N^3)$ arithmetic operations. However, if the encoding vectors are $\omega$-sparse, we can adopt the dual-basis approach in obtaining $\mathbf{B}_k$ as in Appendix~\ref{app:B}, and guarantee that each $\mathbf{B}_k$ can be obtained in $O(\omega N^2)$ times. The computational complexity of Step 1 is thus $O(\omega K  N^2)$. Step~2 involves $O(KN^2)$ operations. In step 3, the {\sc Min HittingSet} problem shown in Theorem~\ref{th:levin_reduction} has a complexity of $O(1.23801^{(N+K)})$\cite{Shianexact2011}. Step~4 requires $O(K)$ operations. Step~5 involves the SA algorithm, which has a complexity of $O(K |\mathcal{H}|)$. Since $|\mathcal{H}| \leq N$, the overall complexity of OH is $O(\omega KN^2+1.23801^{(N+K)}) = O(1.23801^{(N+K)})$. The only difference between OH and GH is that GH uses a greedy algorithm to approximate the {\sc Min HittingSet} problem in Step 3. The greedy algorithm takes $O(KN^2)$ operations. Therefore, the overall complexity of GH is $O(\omega K N^2)$.



\subsection{Solving Binary Equation Set for $q=2$}

The last step of the GH method involves solving a set of linear inequalities over $GF(q)$. However when $q=2$, solving a linear inequality of the form $f(\mathbf{x}) \neq 0$ is equivalent to solving the linear equation $f(\mathbf{x}) = 1$. Based on this fact, we now propose a procedure which is called {\em Solving Binary Equation Set} (SBES), which modifies the SA algorithm so that it is applicable to the case where $q=2$. Note that the same idea can be applied to cases where $q$ is a prime power satisfying $2 < q < K$.

The heuristic is as follows. We want to find $\mathbf{x}$ such that $\mathbf{A}\mathbf{x} = \bone$, where $\mathbf{A}$ is the coefficient matrix of the system of linear equations. The system may be inconsistent and has no solution. Nevertheless, we can disregard some equations and guarantee that at least $rank(\mathbf{A})$ equations are satisfied.

\medskip

\begin{flushleft}
{\bf Solving Binary Equation Set Procedure (SBES):}
\end{flushleft}
\smallskip

\noindent {\bf Input:} A $K \times N$ matrix $\hat{\mathbf{B}}$ over $GF(2)$ and $\mathcal{N}$, where  $\mathcal{N} \subseteq \{1,2,\ldots,N\}$.

\noindent {\bf Output:} $\mathbf{x} = (x_1, x_2, \ldots, x_N) \in GF(2)^{N}$ with support in $\mathcal{N}$.

\smallskip

\noindent {\bf Step 0:} Let $\mathbf{z}=(z_1, z_2, \ldots, z_{|\mathcal{N}|})$ be the zero vector.


\noindent {\bf Step 1:} Delete columns of $\hat{\mathbf{B}}$ whose column indices are not in $\mathcal{N}$. Augment the resulting matrix by adding a $K$-dimensional all-one column vector to the right-hand side. Let the resulting matrix be denoted by~$\mathbf{Q}$.

\noindent {\bf Step 2:} Compute the row echelon form (REF) of $\mathbf{Q}$ and call it  $\mathbf{Q}'$.

\noindent {\bf Step 3:} Delete all zero rows in $\mathbf{Q}'$ and any row in $\mathbf{Q}'$ if it has a single ``1'' in  the $(|\mathcal{N}|+1)$-th entry. The resulting matrix is called $\mathbf{Q}''$. Let the number of pivots in $\mathbf{Q}''$ be $\nu$, and let
$p_1$,  $p_2,\ldots p_\nu$ be the column indices of the pivot in $\mathbf{Q}''$  listed in ascending order.

\noindent {\bf Step 4:} Execute elementary row operations in $\mathbf{Q}''$ so that $\mathbf{Q}''$ is transformed into its row-reduced echelon form.

\noindent {\bf Step 5:} Set the variables associated with the non-pivot columns to zero.
For $i=1,2,\ldots, \nu$,  assign $z_{p_i}$ the value of the $i$-th entry of the last column in $\mathbf{Q}''$.


\noindent {\bf Step 6:} Assign values to the components of $\mathbf{x}$ such that $\mathbf{x}(\mathcal{N}) = \mathbf{z}$, and $x_i=0$ if $i\not\in\mathcal{N}$.

\medskip

When applying the GH method to the case where $q=2$, we replace the SA algorithm in Step~5 of the GH method by the SBES procedure. We call this modification GH with SBES.


\begin{myexam}
Consider $q=2$, $K=4$, $N=5$, $\mathcal{H}=\{1,3\}$ and
\begin{equation}
\hat{\mathbf{B}} = \begin{bmatrix} \hat{\mathbf{b}}_1^T \\ \hat{\mathbf{b}}_2^T \\ \hat{\mathbf{b}}_3^T \\ \hat{\mathbf{b}}_4^T \end{bmatrix} = \begin{bmatrix} 1&1&0&1&0 \\ 1&1&1&0&1 \\ 1&0&0&1&1 \\ 0&0&1&0&0 \end{bmatrix}.
\end{equation}
We extract the first and third rows of $\hat{\mathbf{B}}$
and augment it by the all-one column vector, $\bone$,
\begin{align*}
\mathbf{Q} &= \begin{bmatrix} 1&0&1\\ 1&1&1\\ 1&0&1\\ 0&1&1\\ \end{bmatrix}.
\end{align*}
In Step 2, we compute the REF of $\mathbf{Q}$
\begin{align*}
\mathbf{Q}' &= \begin{bmatrix} 1&0&1\\ 0&1&0\\ 0&0&1\\ 0&0&0\\ \end{bmatrix}.
\end{align*}
The last row is an all-zero row, representing a redundant equation. The second last row contains a one in the third component, and zero elsewhere, implying that the original system of linear equations cannot be solved. In order to get a heuristic solution, we relax the system by deleting that row, and obtain
\begin{align*}
\mathbf{Q}'' &= \begin{bmatrix} 1&0&1\\ 0&1&0\\  \end{bmatrix}.
\end{align*}
We have $p_1=1$ and $p_2=2$ in Step 3. The matrix $\mathbf{Q}''$ is already in its RREF. By Step 5, $z_{1} = q''(1,3)=1$ and $z_{2} = q''(2,3)=0$. Finally, we set $x_{1} = z_{1} = 1$, and $x_{3}= z_{2}= 0$. As a result, $\mathbf{x}$ equals $(1,0,0,0,0)$ and is innovative to all except the last user. \hfill $\Box$
\end{myexam}

In the above example, an innovative vector does exist, but GH with SBES fails to find it. The main reason is that the hitting set subproblem is formulated for the case where $q \geq K$. When $q$ is small, there is no guarantee that a non-empty $\mathcal{I}$ must consist of a vector with support restricted in $\mathcal{H}$. Indeed, we will skip the greedy hitting procedure and simply let $\mathcal{H}$ be the index set of all packets, i.e., $\{1, 2, \ldots, N\}$, then it is easy to check that with the same input, the SBES procedure returns $\mathbf{x}=(0,0,1,1,0)$, which is an innovative vector to all users. We call this modification {\em Full Hitting} (FH) with SBES, for the reason that the hitting set is chosen as the full index set of the packets. In general, FH with SBES produces encoding vectors that are innovative to more users than GH with SBES, at the expense of higher Hamming weights. Note that the encoding vectors generated by GH with the SBES procedure may not be innovative when $K > q = 2$. They are still innovative, however, to a fraction of the $K$ users.

Now we analyze the computational complexity of FH with SBES and GH with SBES. The SBES procedure is indeed the Gauss-Jordan elimination. So the complexity of SBES is $O(NK^2)$. The only difference between GH method and GH with SBES is that GH with SBES uses SBES method in the last step, instead of the SA algorithm. Therefore, the total complexity of GH with SBES is $O(\omega K N^2 + NK^2)$. For FH with SBES, compared with GH with SBES, it skips the greedy hitting procedure which has a complexity of $O(KN^2)$. Since $\omega \geq 1$, the overall complexity of FH with SBES is also $O(\omega K N^2 + NK^2)$.

\section{Benchmarks for Comparison} \label{sec:benchmark}

In this section, we describe four existing linear broadcast codes, which will be used for comparison. They are LT code\cite{LT}, Chunked Code\cite{lithreeschemes2012}, Random Linear Network Code (RLNC), and Instantly Decodeable Network Code (IDNC)\cite{Sorourcompletion2012}. Their encoding and decoding complexities will also be presented.

\subsection{Code Descriptions}

When LT code is used, the base station repeatedly broadcasts encoded packets to all users until they have decoded all the source packets. The encoding vectors are generated as follows. The base station first randomly picks a degree value, $d$, according to the \emph{Robust Soliton} distribution (see Definition 11, \cite{LT}). It then selects $d$ source packets uniformly at random and generates an encoded packet by adding them over $GF(2)$. A belief propagation decoder is used for decoding. For more details about the LT code, we refer the readers to \cite{LT}.

RLNC can be used with two transmission phase. In the systematic phase, the base station broadcasts each source packet once. In the retransmission phase, it repeatedly broadcasts encoded packets generated over $GF(q)$ to all users until they have decoded all the source packets. Gaussian elimination is used for decoding. A user sends an acknowledgement to the base station after it has decoded all the source packets.

Chunked Code is an extension of RLNC. It divides the broadcast packets into disjoint chunks with $C$ packets per chunk. For simplicity, we assume $C$ divides $N$. At each time slot, the base station first picks a chunk uniformly at random. Then, it generates an encoded packet by linearly combining the packets in the selected chunk with coefficients drawn from $GF(q)$. As a result, all the encoding vectors are $C$-sparse. A user can decode a chunk as long as it has received $C$ linearly independent encoded packets which are generated from that chunk. Gaussian elimination is used for decoding. A user sends an acknowledgement to the base station after it has decoded all the source packets.

IDNC has many different variations. We consider the Maximum Weight Vertex Search (MWVS) algorithm proposed in \cite{Sorourcompletion2012}. The broadcast is divided into two phases, the systematic phase and the retransmission phase. In the systematic phase, each source packet is broadcast to all users once. At each time slot of the retransmission phase, the base station generates an encoded packet using MWVS, and broadcasts it to all users. Each user sends a feedback to the base station after receiving the packet. Then the base station generates a new encoded packet based on the updated feedback. Such a procedure is repeated until all users have decoded all the source packets. The encoded packets are generated as follows. After the systematic phase, if the $i$-th user has not received the $j$-th source packet, a vertex $v_{ij}$ will be defined. Two vertices, say $v_{ij}$ and $v_{kl}$, are connected if
\begin{enumerate}
\item $i\neq k$ and $j=l$; or
\item $i\neq k$, $j\neq l$, the $i$-th user has the $l$-th packet and the $k$-th user has the $j$-th packet.
\end{enumerate}
By doing so, a graph $G = (\mathcal{V}, \mathcal{E})$, called the IDNC graph, is constructed. Its adjacency matrix is defined in the usual way, that is, if vertex $v_{ij}$ is connected to vertex $v_{kl}$, the element $a_{ij,kl}$ will be set to 1, otherwise 0. Then, define the vertex weight $w_{ij}$ for each vertex $v_{ij}$ as:
\begin{equation} \label{eq:vertex_weight}
w_{ij} \triangleq \frac{\tau_{i}}{1-p_i} \big(\sum_{k,l : v_{kl} \in \mathcal{V}} \frac{\tau_{k}}{1-p_k}a_{ij,kl} \big),
\end{equation}
where $p_i$ is the erasure probability\footnote{In calculating the vertex weight, it was assumed that the transmission from the base station to user~$i$ experiences independent erasures in different time slots, each with probability $p_i$} of the channel from the base station to user $i$, and $\tau_{i}$ is the number of source packets not yet received by the $i$-th user. The vertex with the maximal weight will then be determined and added to an initially empty set, $\mathcal{K}$. The following iterative procedure will then be repeatedly performed:
\begin{enumerate}
\item Let $\mathcal{V}' \triangleq \{ v \in \mathcal{V} \setminus \mathcal{K} : (v, u) \in \mathcal{E} \; \forall u \in \mathcal{K}\}$.
\item Let $G'$ be the subgraph induced by $\mathcal{V}'$.
\item Recompute the vertex weights using~\eqref{eq:vertex_weight} with respect to the graph $G'$.
\item Find the vertex in $G'$ with maximal weight and adds it into $\mathcal{K}$.
\end{enumerate}
The above procedure is repeated until no more vertices in $G$ can be added to $\mathcal{K}$, which occurs when $\mathcal{V}'$ in Step~1 is found to be empty. Note that the vertices in $\mathcal{K}$ form a clique in $G$. An encoded packet is obtained by XORing all the packets in $\{P_j : v_{ij} \in \mathcal{K} \text{ for some } i \}$.


\subsection{Complexity Analysis}

We first consider the encoding complexity. For LT code, it first generates the degree value $d$ according to the Robust Soliton distribution. This can be done by generating a random variable uniformly distributed over $(0,1)$ and checking which interval it falls within, out of the $N$ intervals defined by the Robust Soliton distribution. Encoding is completed by randomly picking $d$ source packets and adding them together over GF(2). The overall complexity is $O(N)$. For RLNC, the encoding complexity is also $O(N)$, since it requires the generations of $N$ random coefficients, $N$ multiplications and $N-1$ additions, all over GF($q$). Chunked Code is a special case of RLNC, which applies RLNC to each chunk. Since each chunk has only $C$ packets, the encoding complexity is $O(C)$. For MWVS, the IDNC graph can have at most $KN$ vertices. To construct the graph, it takes $KN(KN-1)$ operations to check the connectivity for each pair of vertices. For each iteration, it takes $O(KN)$ operations to compute the vertex weight for each vertex. As aforementioned, the IDNC graph can have at most $KN$ vertices, therefore the complexity of each iteration is $O(K^2N^2)$. Since the vertex set $\mathcal{K}$ can have at most $K$ vertices, the algorithm can run at most $K$ iterations. The overall encoding complexity of MWVS is $O(K^3 N^2)$. For more details about the algorithm, we refer the readers to \cite{Sorourcompletion2012}. As analyzed in previous sections, the encoding complexity of GH is $O(\omega K N^2)$. Since it is known that $\omega \leq K$, the encoding complexity of GH is lower than that of MWVS.

Now we are going to analyze the decoding complexity. For LT code, we use belief propagation decoder for decoding. The decoding process and iterative and starts from packets with degree~1. Note that the received packets and their corresponding encoding vectors will be updated during the iterative process. In each iteration, the decoder selects a successfully received packet $Y_i$ whose current encoding vector has Hamming weight~1. In other words, $Y_i$ is equal to $P_k$ for some $k$. Then it subtracts $P_k$ (over GF(2)) from each of the received packets whose $k$-th component of the encoding vector is equal to 1.
The encoding vectors of these packets are then updated accordingly. Such a procedure is repeated until all the source packets are decoded. In the worst case, only one source packet is decoded in each iteration. There will be $N-1$ iterations and the $i$-th iteration requires $N-i$ operations. Totally, $\frac{N(N-1)}{2}$ operations are needed. Therefore, the decoding complexity of LT code is $O(N^2)$. For RLNC, we assume that Gaussian elimination is used for decoding. The decoding complexity is therefore $O(N^3)$. For the Chunked Code, the decoding complexity to solve each chunk is $O(C^3)$. Totally, there are $\lceil N/C \rceil$ chunks. Therefore, its decoding complexity is $O(C^2N)$. For GH, OH, GH with SBES, FH with SBES, since all the encoding vectors are guaranteed to be $K-$sparse, decoding algorithms for solving sparse linear system can be used~\cite{sungonthesparsity2011}. The decoding complexity is  $O(\min\{K,N\} N^2)$. For IDNC, since $\mathcal{K}$ has at most $K$ vertices, the encoding vector of an encoded packet cannot have Hamming weight greater than $\min\{K,N\}$. Therefore, the decoder needs at most $\min\{K,N\}-1$ XOR operations to decode a source packet. Since there are totally $N$ source packets, its decoding complexity is $O(\min\{K,N\} N)$.

The encoding and decoding complexities of these schemes are summarized in Table~\ref{ta:complexity}. Since $\omega$ in general grows with $K$, in the table, we replace all $\omega$ by $K$, as we know that $\omega$ is always bounded above by $K$.


\begin{table}
\centering
\begin{threeparttable}[b]
\caption{Comparison of Computational Complexity}
\label{ta:complexity}
\renewcommand{\arraystretch}{1.3}
\begin{tabular}{|p{1.8cm}||p{3cm}|p{2.3cm}|}
\hline
\bfseries Scheme & \bfseries Encoding & \bfseries{Decoding} \\
\hline
\bfseries LT Code & $O(N)$ & $O(N^2)$ \\
\hline
\bfseries RLNC & $O(N)$ &  $O(N^3)$ \\
\hline
\bfseries Chunked Code & $O(C)$ & $O(C^2N)$ \\
\hline
\bfseries IDNC (using MWVS) & $O(K^3N^2)$ & $O(\min\{K,N\}N)$ \\
\hline
\bfseries OH & $O(1.23801^{(N+K)})$ & $O(\min\{K,N\} N^2)$\\
\hline
\bfseries GH & $O(K^2 N^2)$ & $O(\min\{K,N\} N^2)$ \\
\hline
\bfseries GH with SBES & $O(K^2 N^2)$ & $O(\min\{K,N\} N^2)$ \\
\hline
\bfseries FH with SBES & $O(K^2 N^2)$ & $O(\min\{K,N\} N^2)$ \\
\hline
\end{tabular}
    \end{threeparttable}
\end{table}

\section{Performance Evaluation} \label{sec:PE}

\subsection{Simulation Setup}

In this section, we evaluate our proposed methods via simulations. We simulate a broadcast system in which a transmitter broadcasts $N$ equal-size packets to $K$ users via erasure broadcast channels. Packet erasures are assumed to be statistically independent across users and time slots. Each user experiences packet erasure with probability $P_e$. We consider both the case where the transmitter receives perfect feedback from the users and the case where each feedback channel has an erasure probability $P_{e_{up}}$.

The whole broadcast process is divided into two phases. In the first phase, the transmitter sends all source packets one by one without coding. In the second phase, packets are encoded and transmitted until all users received enough packets for recovering all the source packets. The number of packets that have been transmitted is equal to the maximum download delay of all users. We call it the {\em completion time} of the broadcast process.

We compare the performance of our proposed codes and the benchmarks described in the previous subsection. For the LT code,
we let $c=0.1$ and $\delta=0.1$ for the Robust Soliton distribution. For the Chunked Code, we let $C = 8$. For our proposed methods, OH, GH and FH, we use Gaussian elimination for decoding. The performance of each code is evaluated in four ways, namely, completion time, encoding time, Hamming weights of encoding vectors, and decoding time. The completion time, as defined above, is a network performance measure, which indicates how efficient the wireless spectrum is utilized. The encoding time reflects the computational complexity of each coding algorithm. It is measured by the CPU time each coding algorithm takes in generating an encoded packet in the second phase. The Hamming weights of the encoding vectors in the second phase are recorded, so that we can evaluate the sparsity of our codes. The decoding time reflects the computational complexity of each decoding algorithm. It is important to mobile devices that do not have very high computational speed. It consists of two parts. The first part is the CPU time that is used for checking whether an encoding vector is innovative to each user. The second part is the CPU time used for decoding computation.


The system configuration of the computer used for the simulation is shown below:
\begin{itemize}
\item Intel\textsuperscript{\circledR} Core\textsuperscript{TM} 2 Quad CPU Q9650 processor (3.00GHz)
\item Microsoft\textsuperscript{\circledR} Windows\textsuperscript{\circledR} 7 Enterprise (64-bit) with Service Pack 1
\item 4GB of RAM
\end{itemize}
All the simulation programs are written in C programming language.

To ensure the accuracy of the measured CPU time, we specify only one CPU core for the simulation process. The priority level of the simulation process is set to the highest possible value, so that the simulation will not be interrupted by any other processes. Then we use the \emph{QueryPerformanceCounter} command to measure the CPU time, which can achieve microsecond precision. In our simulation figures, each data point involves 3,000 random realizations.

\subsection{Simulation Results}

We first evaluate the performance over $GF(2^8)$ with various~$N$. We let $K=40$, $P_e=0.3$ and $N$ varying from 32 to 96. The uplink channels are assumed to be perfect, i.e., $P_{e_{up}}=~0$. The average completion times of all the coding schemes are shown in Fig.~\ref{fig_CompletionDelay_vs_N}. Both OH and GH are optimal, since they always generate innovative vectors for $q \geq K$. Compared with LT code and Chunked Code, when $N=96$, they can reduce their completion times by 74\% and 50\%, respectively. Since $GF(2^8)$ can be considered sufficiently large for $K=40$, encoding vectors generated by RLNC are almost always innovative. As a result, RLNC performs almost the same as OH and GH.

The simulation results of encoding time is shown in Fig.~\ref{fig_encoding_time_vs_N}. Since OH, IDNC and GH need to perform some computation based on the feedback from users, they have longer encoding times than LT code, Chunked Code and RLNC. While OH has longer encoding time than IDNC, GH is able to outperform IDNC by 21\% when $N=96$. Although both RLNC and LT have linear encoding complexity, LT works over $GF(2)$, which is much faster than RLNC in terms of CPU time. It can also be seen that the growth of encoding time with increasing $N$ for each scheme agrees with its encoding complexity analyzed in the previous section.

\begin{figure}
\centering
\includegraphics[width=3.75in]{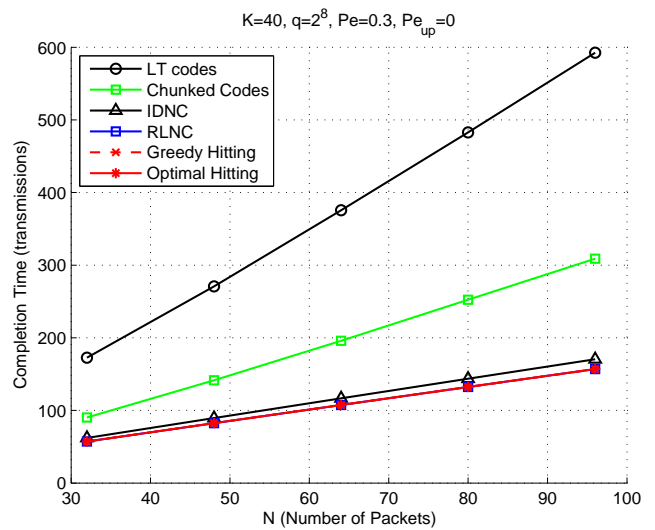}
\caption{Completion time for various $N$ with $q=2^8$, $K=40$, $P_e=0.3$  and  $P_{e_{up}}=0$.}
\label{fig_CompletionDelay_vs_N}
\end{figure}

\begin{figure}
\centering
\includegraphics[width=3.75in]{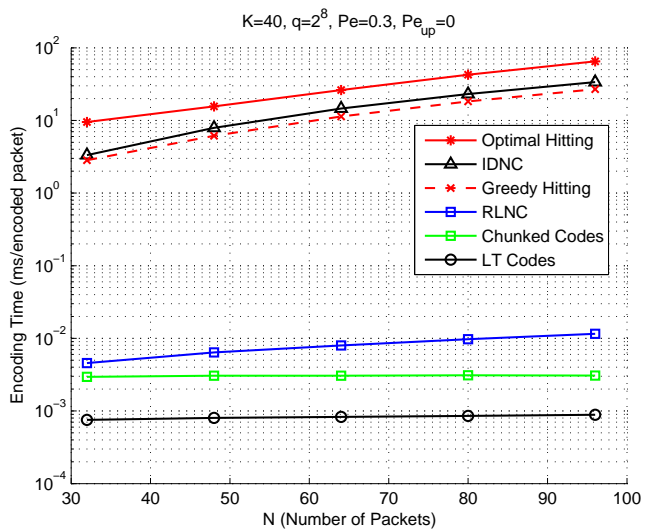}
\caption{Encoding time for various $N$ with $q=2^8$, $K=40$, $P_e=0.3$  and  $P_{e_{up}}=0$.}
\label{fig_encoding_time_vs_N}
\end{figure}

The comparisons of Hamming weight and decoding time are shown in Fig.~\ref{fig_HW_vs_N} and \ref{fig_decoding_time_vs_N}, respectively. The Hamming weight of encoding vectors generated by OH and GH is comparable to that of LT code and IDNC, while much smaller than that of RLNC. The advantage of transmitting sparse encoding vectors is that the decoding time can be reduced, which is shown in Fig.~\ref{fig_decoding_time_vs_N}. Fig.~\ref{fig_decoding_time_vs_N} shows that the decoding time of RLNC is much higher than those of the other schemes. In particular, compared with RLNC, both OH and GH can reduce decoding time by 97\% when $N=96$. Besides, the growth of decoding time with increasing $N$ for each scheme agrees with the complexity analysis in the previous section.


\begin{figure}
\centering
\includegraphics[width=3.75in]{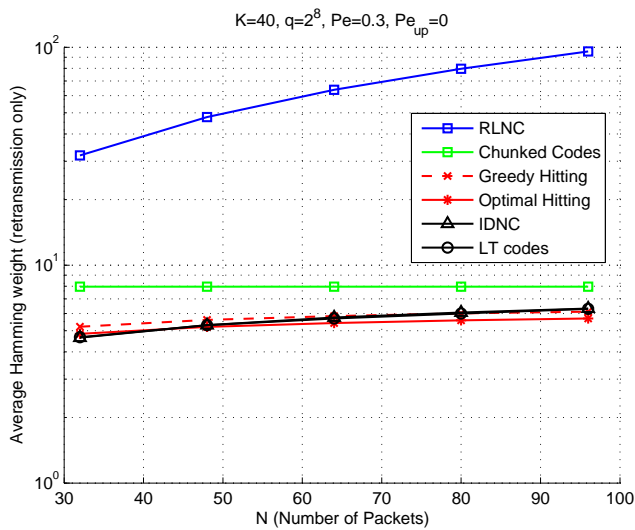}
\caption{Hamming weight for various $N$ with $q=2^8$, $K=40$, $P_e=0.3$  and  $P_{e_{up}}=0$.}
\label{fig_HW_vs_N}
\end{figure}

\begin{figure}
\centering
\includegraphics[width=3.75in]{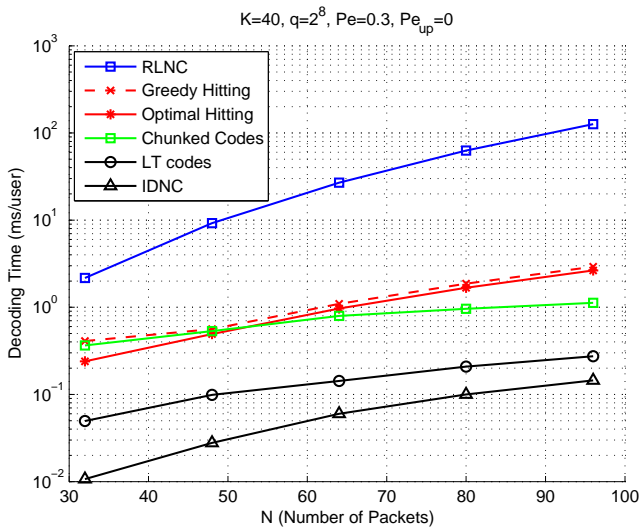}
\caption{Decoding time for various $N$ with $q=2^8$, $K=40$, $P_e=0.3$  and  $P_{e_{up}}=0$.}
\label{fig_decoding_time_vs_N}
\end{figure}

Next, we investigate the performance with various $K$. We let $N=32$, $P_e=0.3$ and $K$ varying from 5 to 200 over $GF(2^8)$. The uplink channels are also assumed to be perfect. The comparison of completion time is shown in Fig.~\ref{fig_CompletionDelay_vs_K}. OH and GH are optimal and outperforms LT code, Chunked Code and IDNC. In particular, compared with LT code, Chunked Code and IDNC, both OH and GH can reduce completion time up to 68\%, 40\% and 17\%, respectively. The gap between the performance of IDNC and OH becomes larger as $K$ increases. The reason is that IDNC works over $GF(2)$ and cannot guarantee innovative encoding vectors when $K$ is large. The comparison of encoding time is shown in Fig.~\ref{fig_encoding_time_vs_K}. It can be seen that the encoding time of IDNC grows quickly with $K$, when compared with OH and GH; the curves for RLNC, LT code and Chunked Code are flat, since their encoding methods are independent of the number of users. All these phenomena agrees with the complexity results shown in the previous section.

 The average Hamming weights of the encoding vectors generated by RLNC, LT code and Chunked Code are constant, as shown in Fig.~\ref{fig_HW_vs_K}. Since the decoding time mainly depends on the Hamming weight of received encoding vectors, the decoding times of RLNC, LT code and Chunked Code keep unchanged as $K$ increases. However, the decoding times of OH, GH and IDNC slightly increase as $K$ increases, which is shown in Fig.~\ref{fig_decoding_time_vs_K}. This agrees with the complexity results shown in the previous section.

\begin{figure}
\centering
\includegraphics[width=3.75in]{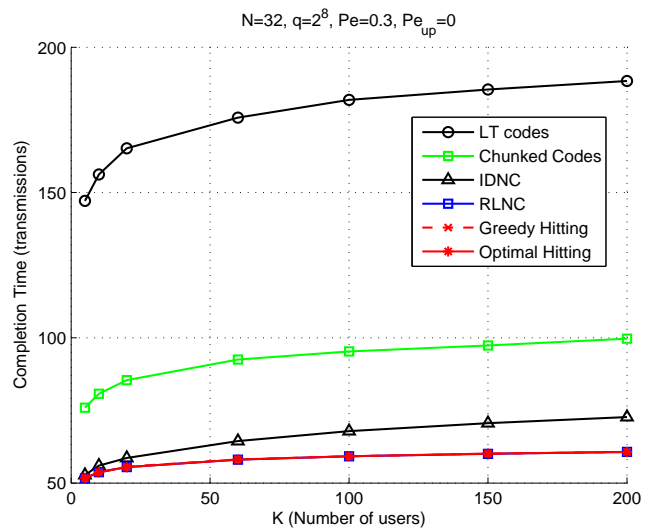}
\caption{Completion time for various $K$ with $q=2^8$, $N=32$, $P_e=0.3$  and  $P_{e_{up}}=0$.}
\label{fig_CompletionDelay_vs_K}
\end{figure}

\begin{figure}
\centering
\includegraphics[width=3.75in]{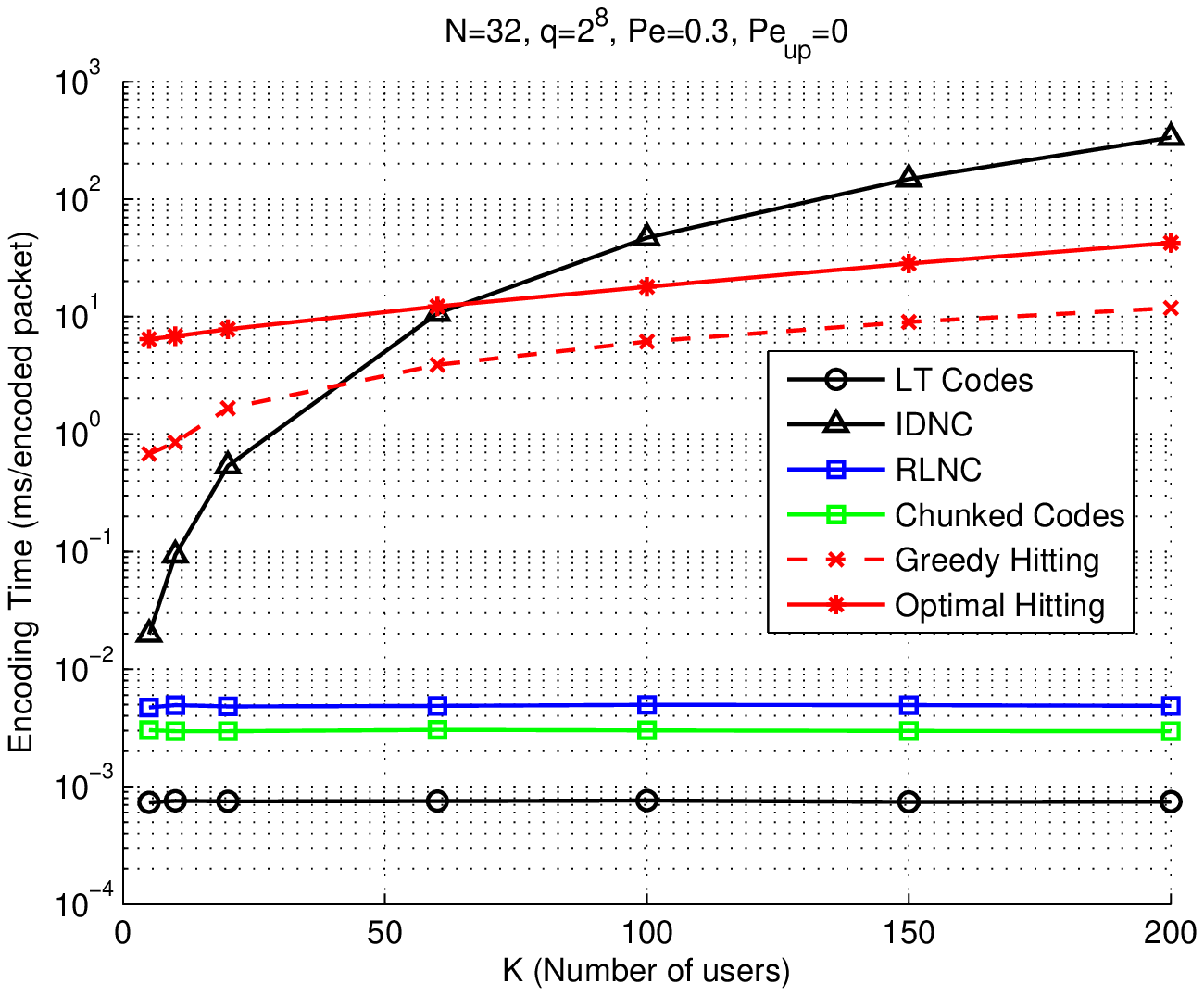}
\caption{Encoding time for various $K$ with $q=2^8$, $N=32$, $P_e=0.3$  and  $P_{e_{up}}=0$.}
\label{fig_encoding_time_vs_K}
\end{figure}

\begin{figure}
\centering
\includegraphics[width=3.75in]{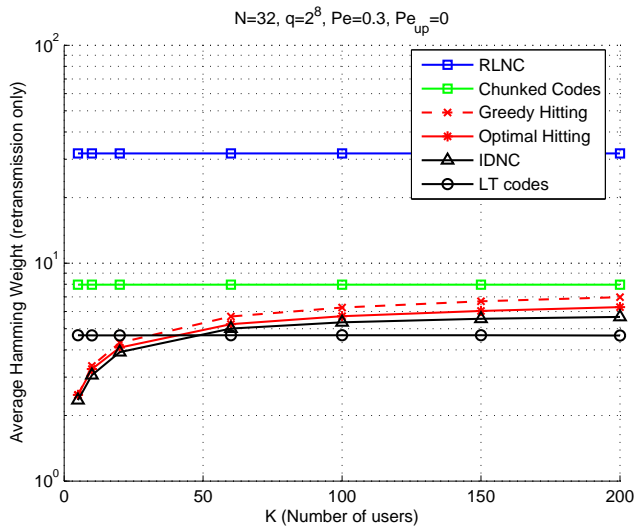}
\caption{Hamming weight for various $K$ with $q=2^8$, $N=32$, $P_e=0.3$  and  $P_{e_{up}}=0$.}
\label{fig_HW_vs_K}
\end{figure}

\begin{figure}
\centering
\includegraphics[width=3.75in]{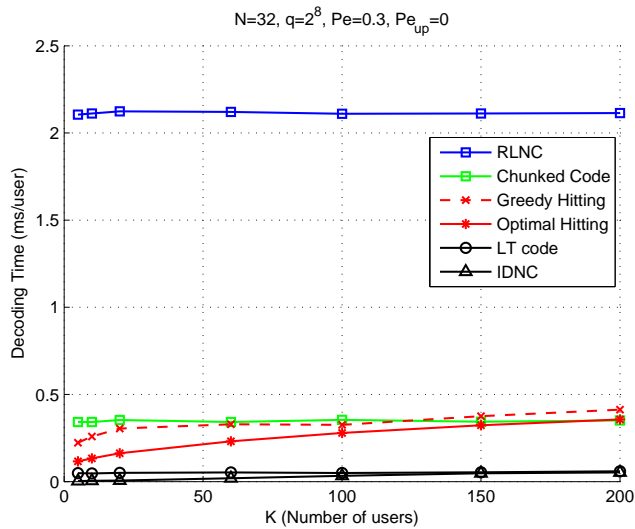}
\caption{Decoding time for various $K$ with $q=2^8$, $N=32$, $P_e=0.3$  and  $P_{e_{up}}=0$.}
\label{fig_decoding_time_vs_K}
\end{figure}

To study the performance under different fading environment, we run simulations with various erasure probability $P_e$. We set $N=32$, $K=40$ and $P_{e_{up}}=0$ over $GF(2^8)$. The erasure probability $P_e$ various from $0$ to $0.4$. The completion times of the coding schemes are shown in Fig.~\ref{fig_CompletionDelay_vs_Pe}. Since OH, GH and RLNC always transmit innovative packets, they have shorter completion times in all scenarios. The comparison of encoding times, Hamming weight and decoding times are similar to Fig.~\ref{fig_encoding_time_vs_N}, \ref{fig_HW_vs_N} and \ref{fig_decoding_time_vs_N}, respectively, and are omitted here.

\begin{figure}
\centering
\includegraphics[width=3.75in]{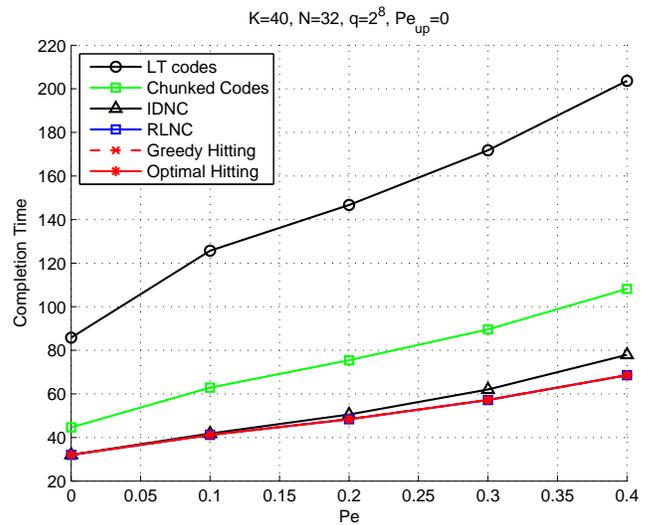}
\caption{Completion time for various $P_e$ with $q=2^8$, $N=32$, $K=40$ and  $P_{e_{up}}=0$.}
\label{fig_CompletionDelay_vs_Pe}
\end{figure}

The performance comparison of these schemes are summarized in Table~\ref{table_comp_schemes}. It is clear that there is no dominant scheme and there is a tradeoff between different performance measures. Our proposed methods, OH and GH, achieves optimal completion time like RLNC. As RLNC has very long decoding time, we design OH and GH to generate sparse encoding vectors so that the decoding time can be reduced. It is, however, at the expense of higher encoding time and user feedback.

\begin{table}
\caption{Performance Comparison of Different Coding Schemes}
\label{table_comp_schemes}
\renewcommand{\arraystretch}{1.3}
\centering
\begin{tabular}{|p{0.8cm}|p{1cm}|p{1cm}|p{1.2cm}|p{1cm}|p{0.9cm}|}
\hline
\bfseries Scheme & \bfseries Comple-tion Time & \bfseries Decoding & \bfseries Encoding & \bfseries Feedback Required?& \bfseries Field Size $q$\\
\hline\hline
LT Code & Long & Fast & Fast & No & $q=2$\\
\hline
RLNC & Optimal & Very Slow & Fast  & No & $q\gg K$\\
\hline
Chunked Code & Sub-optimal & Faster than RLNC & Fast & No & $q\gg K$\\
\hline
IDNC & Sub-optimal &Fast &Slow & Yes & $q=2$\\
\hline
OH & Optimal & Faster than RLNC and GH & Slow & Yes & $q\geq K$\\
\hline
GH & Optimal & Faster than RLNC& Slow (Faster than OH) & Yes & $q\geq K$\\
\hline
\end{tabular}
\end{table}

Since the binary field is widely used in engineering, we also study the performance of these schemes over $GF(2)$. In the binary field, innovative encoding vectors cannot be guaranteed. However, simulation shows that FH and GH still have shorter completion time. In particular, compared with LT, Chunked Code, IDNC and RLNC, they can reduce completion time by 77\%, 53\%, 15\% and 5\%, respectively, when $K=200$.

\begin{figure}
\centering
\includegraphics[width=3.75in]{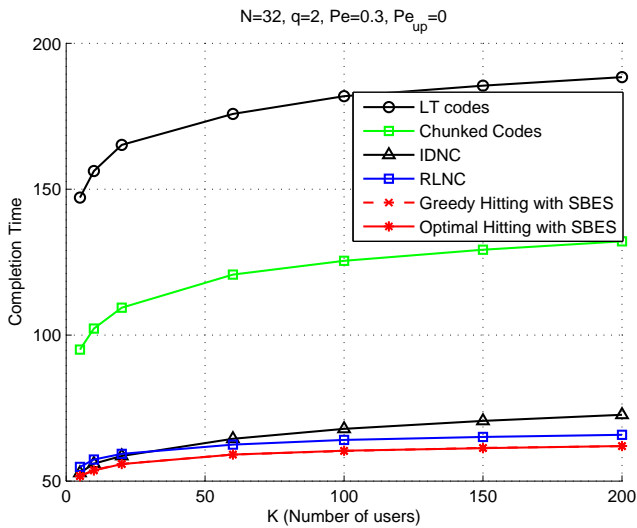}
\caption{Completion time for various $K$ with $q=2$, $N=32$, $P_e=0.3$  and  $P_{e_{up}}=0$.}
\label{fig_Binary_CompletionDelay_vs_K}
\end{figure}

As feedback channels in reality are not error-free, we also investigate the case where the feedback channels have erasures. In Fig.~\ref{fig_CompletionDelay_vs_K_peup0dot1}, we plot the completion times with $N=32$, $P_e=0.3$ and $P_{e_{up}}=0.1$ over $GF(2^8)$. It can be seen that OH and GH have longer completion times than RLNC does, but still outperforms LT code, Chunked Code and IDNC. In particular, when $K=200$, compared with RLNC, OH and GH requires 15\% longer completion time than RLNC, while compared with LT code, Chunked Code and IDNC, they can reduce their completion times by 73\%, 30\% and 20\%, respectively.

\begin{figure}
\centering
\includegraphics[width=3.75in]{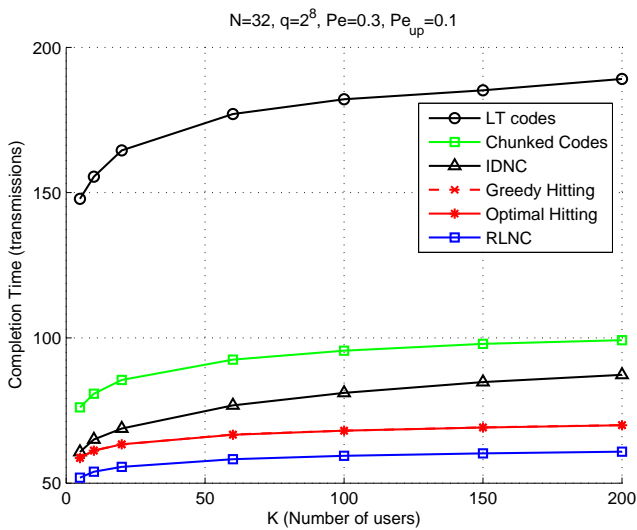}
\caption{Completion time for various $K$ with $q=2^8$, $N=32$, $P_e=0.3$ and $P_{e_{up}}=0.1$.}
\label{fig_CompletionDelay_vs_K_peup0dot1}
\end{figure}

\section{Conclusions} \label{sec:conclude}

In this paper, we adopt the computational approach to study the linear network code design problem for wireless broadcast systems. To minimize the completion time or to maximize the information rate, the concept of innovativeness plays an important role. We show that innovative linear network code is uniformly optimal in minimizing user download delay. While it is well known that innovative encoding vectors always exist when the finite field size, $q$, is greater than the number of users, $K$, we prove that the problem of determining their existence over smaller fields is NP-complete. Its corresponding maximization version is not only hard to solve, but also hard to approximate. Nevertheless, for $GF(2)$, we propose a heuristic called FH with SBES, which is numerically shown to be nearly optimal under our simulation settings.

Sparsity of a network code is another issue we have addressed. When $q \geq K$, we show that the minimum Hamming weight within the set of innovative vectors is bounded above by $K$. To find a vector that achieves the minimum weight, however, is proven to be NP-hard via a reduction from the hitting set problem. An exact algorithm based on BIP is described, and a polynomial-time approximation algorithm based on the greedy approach is constructed.

The performance of our proposed algorithms has been evaluated by simulations. When $q \geq K$, our proposed algorithm is optimal in completion time and is effective in reducing decoding time. When $q=2$, our proposed algorithm is able to strike a proper balance between completion time and decoding time. Our proposed methods, however, require longer encoding time and the availability of user feedback. We believe that there is no scheme which can dominate all other schemes in all aspects. Which broadcast codes are more suitable to use depends on the specific application scenario.

{\em Acknowledgements:} The authors would like to thank Prof. Wai Ho Mow and Dr. Kin-Kwong Leung for their stimulating discussions in the early stage of this work.

\appendices

\section{Problem instances with no Innovative Encoding Vector when $q<K$}
\label{app:A}
In this appendix we show that the condition $q\geq K$ in Theorem~\ref{nonemptyI} cannot be relaxed.

Let $U$ be the ambient space $GF(q)^N$, and consider a subspace $V$ of $U$ with dimension $N-2$.
Let $\mathbf{v}_1, \mathbf{v}_2, \ldots, \mathbf{v}_{N-2}$ be a basis of~$V$. For a given vector $\mathbf{u}$ in $U$, we let $V\oplus \langle \mathbf{u} \rangle$ denote the
the vector subspace in $U$ generated by $V$ and $\mathbf{u}$.
We claim that we can find $K=q+1$
non-zero vectors $\mathbf{u}_1, \mathbf{u}_2,\ldots, \mathbf{u}_K$ in $U$ such that for $i=1,2,\ldots, K$, the sets of vectors
$ (V\oplus \langle \mathbf{u}_i \rangle)\setminus V$
are mutually disjoint. We can pick $\mathbf{u}_1, \mathbf{u}_2,\ldots, \mathbf{u}_K$ sequentially as follows. Firstly, we let $\mathbf{u}_1$ to be any vector in $U\setminus V$. For $1< j \leq K$, we let $\mathbf{u}_j$ to be any vector in
\begin{equation}
 U \setminus \Big( \bigcup_{i=1}^{j-1} (V\oplus \langle \mathbf{u}_i\rangle ) \Big).
 \label{eq:nonempty}
\end{equation}
The set in \eqref{eq:nonempty} is non-empty because, by the union bound, we have
\begin{align*}
 \Big| \bigcup_{i=1}^{j-1} (V\oplus \langle \mathbf{u}_i \rangle) \Big|
 &=  |V| + \Big| \bigcup_{i=1}^{j-1} (V\oplus \langle\mathbf{u}_i\rangle ) \setminus V \Big| \\
 &\leq q^{N-2} + (j-1) (q^{N-1}-q^{N-2}) \\
 &< q^{N-2} + (q+1) (q^{N-1}-q^{N-2}) = |U|.
\end{align*}
If $\mathbf{u}_1, \mathbf{u}_2,\ldots, \mathbf{u}_K$ are chosen according to the above procedure, then $(V\oplus \langle \mathbf{u}_i \rangle)\setminus V$ and $(V\oplus \langle\mathbf{u}_j\rangle)\setminus V$ are disjoint for $i\neq j$. Otherwise, if  $(V\oplus \langle\mathbf{u}_i\rangle )\setminus V$ and $(V\oplus \langle\mathbf{u}_j \rangle)\setminus V$ have non-empty intersection for some $i<j$, then we have
\[
  \alpha \mathbf{u}_i + \mathbf{v} = \alpha' \mathbf{u}_j + \mathbf{v}'
\]
for some non-zero scalar $\alpha$ and $\alpha'$ in $GF(q)$ and vectors $\mathbf{v}$ and $\mathbf{v}'$ in $V$, but this implies that
\[
 \mathbf{u}_j = \frac{1}{\alpha'} \big( \alpha \mathbf{u}_i + \mathbf{v} -\mathbf{v}' \big) \in (V\oplus \langle \mathbf{u}_i \rangle)\setminus V,
\]
contradicting the condition in \eqref{eq:nonempty}.
For $k=1,2,\ldots, K=q+1$, we create an instance of the problem of finding an innovative encoding vector by defining $\mathbf{C}_k$ as the $(N-1) \times N$ matrix whose row vectors are $\mathbf{v}_1, \ldots, \mathbf{v}_{N-2}$, and $\mathbf{u}_k$. The row spaces corresponding to the matrices $\mathbf{C}_k$'s satisfy

\noindent (i) $\text{rank}(\mathbf{C}_k) = N-1$ for all $k$.

\noindent (ii) $\text{rowspace}(\mathbf{C}_i) \cap \text{rowspace}(\mathbf{C}_j) = V$ whenever $i\neq j$.

\noindent  (iii) For $k=1,2,\ldots, q+1$, the sets $\text{rowspace}(\mathbf{C}_k) \setminus V$ are mutually disjoint.

The size of the union of $\text{rowspace}(\mathbf{C}_k)$ is
\begin{align*}
& |V| + \sum_{k=1}^{K} |\text{rowspace}(\mathbf{C}_k) \setminus V| \\
& = q^{N-2} + (q+1)(q^{N-1} - q^{N-2})
  = q^N.
\end{align*}
Hence the rowspaces of $\mathbf{C}_k$'s cover the whole vector space $GF(q)^N$. Any encoding vector we pick from $GF(q)^N$ is not innovative to at least one user.

As an example, we consider the case $q=3$ and $K=4$ and $N=3$. If the encoding matrices are
\[
\begin{bmatrix}
1&1&1\\1&0&0
\end{bmatrix},
\begin{bmatrix}
1&1&1\\0&1&0
\end{bmatrix},
\begin{bmatrix}
1&1&1\\0&0&1
\end{bmatrix},
\begin{bmatrix}
1&1&1\\0&1&2
\end{bmatrix},
\]
then we cannot find any innovative encoding vector.

\section{Incremental Method for Computing a Basis of the Null Space of a Given Matrix}
\label{app:B}
In this appendix, we illustrate how  to compute a basis of the null space {\em incrementally}.
In the application to the broadcast system we consider
in this paper, the rows of $\mathbf{C}$ are given one by one. A row is revealed after an innovative packet is received.
Given an $r\times N$ matrix $\mathbf{C}$ over $GF(q)$, recall that our objective is to find a basis for the null space of~$\mathbf{C}$.
The idea is as follows.
We first extend $\mathbf{C}$ to an $N\times N$ matrix by appending $N-r$ row vectors. These vectors are chosen in a way such that the resulting matrix, denoted by $\tilde{\mathbf{C}}$, is non-singular. Let $\tilde{\mathbf{B}}$ be the inverse of $\tilde{\mathbf{C}}$. By the very definition of matrix inverse, the last $N-r$ columns of $\tilde{\mathbf{B}}$ is a basis for the null space of $\mathbf{C}$.


We proceed by induction. The algorithm is initialized by setting  $\tilde{\mathbf{C}}=\tilde{\mathbf{B}}=\mathbf{I}_{N}$.
We will maintain the property that $\tilde{\mathbf{C}}^{-1} = \tilde{\mathbf{B}}$.

Suppose that the first $r$ rows of $\tilde{\mathbf{C}}$ are the encoding vectors received by a user, and  $\tilde{\mathbf{C}} = \tilde{\mathbf{B}}^{-1}$. We let $\mathbf{c}_{i}^T$ be the $i$-th row of $\tilde{\mathbf{C}}$ and $\mathbf{b}_{j}$ be the $j$-th column of $\tilde{\mathbf{B}}$.
When a packet arrives, we can check whether it is innovative by taking the inner product of the encoding vector of the new packet, say $\mathbf{w}$, with $\mathbf{b}_{r+1}$, $\mathbf{b}_{r+2}, \ldots, \mathbf{b}_N$. According to Lemma~\ref{th:characterization}, it is innovative to that user if and only if one or more of such inner products are non-zero.


Consider the case that $\mathbf{w}$ is innovative. Permute the columns of $\tilde{\mathbf{B}}$, if necessary, to ensure that $\mathbf{w}^T \mathbf{b}_{r+1}\neq 0$. This can always be done, since $\mathbf{w}$ cannot be orthogonal to all the last $N-r$ columns of $\tilde{\mathbf{B}}$.
Permute the rows of $\tilde{\mathbf{C}}$ accordingly, so as to ensure that $\tilde{\mathbf{C}}^{-1} = \tilde{\mathbf{B}}$.

We are going to modify $\tilde{\mathbf{C}}$ by updating its $(r+1)$-st row to $\mathbf{w}^T$. This operation can be expressed algebraically by
\begin{equation} \label{updateC}
\tilde{\mathbf{C}} \longleftarrow \tilde{\mathbf{C}} + \mathbf{e}_{r+1} (\mathbf{w} - \mathbf{c}_{r+1})^T,
\end{equation}
where $\mathbf{e}_{r+1}$ is the column vector with the $(r+1)$-st component equal to 1 and 0 otherwise. The matrix $\mathbf{e}_{r+1} (\mathbf{w} - \mathbf{c}_{r+1})^T$ is a rank-one matrix, with the $(r+1)$-st row equal to $(\mathbf{w} - \mathbf{c}_{r+1})^T$, and 0 everywhere else.
The inverse of $\tilde{\mathbf{C}} + \mathbf{e}_{r+1} (\mathbf{w} - \mathbf{c}_{r+1})^T$ can be computed efficiently by the Sherman-Morrison formula~\cite{NR} \cite[p.18]{HornJohnson},
\begin{align}
 & \phantom{=} (\tilde{\mathbf{C}} + \mathbf{e}_{r+1} (\mathbf{w} - \mathbf{c}_{r+1})^T)^{-1} \notag \\
&= \tilde{\mathbf{C}}^{-1} -
\frac{\tilde{\mathbf{C}}^{-1} \mathbf{e}_{r+1} (\mathbf{w} - \mathbf{c}_{r+1})^T \tilde{\mathbf{C}}^{-1}  }{1+(\mathbf{w} - \mathbf{c}_{r+1})^T \tilde{\mathbf{C}}^{-1} \mathbf{e}_{r+1}} \notag \\
&= \tilde{\mathbf{C}}^{-1} -
\frac{ \mathbf{b}_{r+1} (\mathbf{w} - \mathbf{c}_{r+1})^T \tilde{\mathbf{C}}^{-1}  }{\mathbf{w}^T \mathbf{b}_{r+1}} \notag  \\
&= \tilde{\mathbf{C}}^{-1} -
\frac{ \mathbf{b}_{r+1} (\mathbf{w}^T \tilde{\mathbf{C}}^{-1} - \mathbf{e}_{r+1}^T) }{\mathbf{w}^T \mathbf{b}_{r+1}}. \label{eq:SM}
\end{align}
We have used the facts that $\tilde{\mathbf{C}}^{-1} \mathbf{e}_{r+1} = \mathbf{b}_{r+1}$ and $\mathbf{c}_{r+1}^T\tilde{\mathbf{C}}^{-1} = \mathbf{e}_{r+1}^T$ in the above equations. The denominator of the fraction in~\eqref{eq:SM} is a non-zero scalar by construction, so that division of zero would not occur.

The updating procedure can now be performed. $\tilde{\mathbf{C}}$ is updated according to~\eqref{updateC} and $\tilde{\mathbf{B}}$ is updated as follows:
\begin{equation}
\tilde{\mathbf{B}} \longleftarrow \tilde{\mathbf{B}} -
\frac{ \mathbf{b}_{r+1} (\mathbf{w}^T \tilde{\mathbf{B}} - \mathbf{e}_{r+1}^T) }{\mathbf{w}^T \mathbf{b}_{r+1}}.
\end{equation}
Note that if  $\mathbf{w}$ is $\omega$-sparse, the multiplication of $\mathbf{w}^T$ and $\tilde{\mathbf{C}}^{-1}$ in~\eqref{eq:SM} can be done in $O(\omega N)$ times.

\bibliographystyle{IEEEtran}



\end{document}